\documentclass[preprint,12pt,authoryear]{article}
\usepackage{anysize}
\usepackage[latin9]{inputenc}
\usepackage{amsmath}
\usepackage{amsthm}
\usepackage{amssymb}
\usepackage{dsfont}
\usepackage{esint}
\usepackage[Large]{subfigure}
\usepackage{graphicx}
\usepackage{rotating}
\usepackage{subfig}
\usepackage[authoryear]{natbib}
\usepackage{color}
\usepackage{lmodern}
\usepackage{amssymb,amsmath}
\usepackage{ifxetex,ifluatex}
\usepackage{fixltx2e}

\makeatletter
\theoremstyle{plain}

\usepackage{anysize}
\usepackage[Large]{subfigure}
\usepackage{graphicx}
\usepackage{lscape}
\usepackage{float}
\usepackage[title,titletoc,toc]{appendix}
\usepackage{color}
\usepackage{blindtext}
\usepackage{comment}
\usepackage{tcolorbox}
\usepackage{color}      
\usepackage{graphicx}
\usepackage{adjustbox}
\usepackage{fancybox}
\usepackage{mathptmx,latexsym}
\usepackage{amssymb,amsmath}
\restylefloat{table}
\newtheorem{theorem}{Theorem}

\newtheorem{definition}[theorem]{Definition}

\newtheorem{proposition}[theorem]{Proposition}
\newtheorem{remark}[theorem]{Remark}

\makeatother

\providecommand{\theoremname}{Theorem}
\providecommand{\keywords}[1]{\textbf{\textit{Keywords: }} #1}

\title{Multivariate MixedTS distribution}

\author{Asmerilda Hitaj\footnote{Department of Statistics and Quantitative Methods, University of Milano-Bicocca.
Email: asmerilda.hitaj1@unimib.it}, Friedrich Hubalek\footnote{Vienna University of Technology, Austria. E-mail: fhubalek@fam.tuwien.ac.at} , Lorenzo Mercuri\footnote{Department of Economics, Management and Quantitative Methods. University
of Milan. CREST Japan Science and Technology Agency. E-mail: lorenzo.mercuri@unimi.it.} and Edit Rroji\footnote{Department of Economics, Business, Mathematical and Statistical Sciences.
University of Trieste. E-mail: erroji@units.it.}}

\begin{document}
\global\long\def\sgn{\mathop{\mathrm{sgn}}}

\maketitle
\begin{abstract}
The multivariate version of the Mixed Tempered Stable
is proposed. It is a generalization of the Normal Variance Mean Mixtures.
Characteristics of this new distribution and its capacity in fitting tails and  capturing dependence structure between components are investigated. We discuss a random number generating procedure and introduce an estimation methodology based on the minimization of a distance between empirical and theoretical characteristic functions. Asymptotic tail behavior of the univariate Mixed Tempered Stable is exploited in the estimation procedure in order to obtain a better model fitting. Advantages of the multivariate Mixed Tempered Stable distribution are discussed and illustrated via simulation study. 
\end{abstract}

\section{Introduction}
The Mixed Tempered Stable (MixedTS from now on) distribution has been introduced in \cite{EditMixedTS2014} and used for portfolio selection in \cite{HMR2015} and for option pricing in \cite{Mercuri2016}. It is a generalization of the Normal Variance Mean Mixtures \citep[see][]{barndorff1982normal} since the structure is similar but its definition generates a dependence of higher moments on the parameters of the standardized Classical Tempered Stable  \citep[see][]{KT2013, KIM} that replaces the Normal distribution.\newline
Recently different multivariate distributions have been introduced in literature for modeling the joint dynamics of  financial time series. For instance, \cite{kaishev2013levy} considers the LG distribution defined as a linear combination of independent Gammas for the construction of a multivariate model whose properties are investigated based on its relation with multivariate splines. Another model is based on the multivariate Normal Tempered Stable  distribution, defined in \cite{bianchi2015riding} as a Normal Mean Variance Mixture with a univariate Tempered Stable distributed mixing random variable that is shown to capture the main stylized facts of multivariate financial time series of equity returns.\newline 
In this paper, we present the multivariate MixedTS distribution and discuss its main features.
The dependence structure in the multivariate MixedTS is controlled by the components of the mixing random vector. A similar approach has been used in \cite{semeraro2008multivariate} for the construction of a multivariate Variance Gamma distribution starting from the idea that the components in the mixing random vector are Gamma distributed. However, as observed in \cite{hitaj2013hedge, hitaj2013J}, Semeraro's model seems to be too restrictive for describing the joint distribution of asset returns. In particular, the sign of the skewness of the marginal distributions determines the sign of the covariance. This means, for instance, if two marginals of a multivariate Variance Gamma have negative skewness, their correlation can never be negative. The additional parameters in the multivariate MixedTS introduce more flexibility in the dependence structure and overcome these limits. Indeed, we compute higher moments for the multivariate MixedTS and show how the tempering parameters break off the bond between skewness and covariance signs.  
\newline   
We discuss a simulation method and  propose an estimation procedure for the multivariate MixedTS. In particular the structure of univariate and multivariate MixedTS allows us to generate trajectories of the process using algorithms that already exist in literature on the simulation of the Tempered Stable distribution (see \cite{KIM}). \newline   
The proposed estimation procedure is based on the minimization of a distance between empirical and theoretical characteristic functions. As explained for instance in \cite{Yu04}, in absence of an analytical density function, estimation based on the characteristic function is a good alternative to the maximum likelihood approach. 
An estimation procedure can be based on the determination of a discrete grid for the transform  variable and on the Generalized Method of Moments (GMM) as for instance in \cite{feuerverger1981efficiency}. The main advantage of this approach is the possibility of obtaining the standard error for estimators whose efficiency increases as the grid grows finer. However, the covariance matrix of moment conditions becomes singular when the number of points in the grid exceeds the sample size. The GMM objective function  explodes thus the efficient GMM estimators can not be computed. To overcome this problem \cite{carrasco2000generalization} developed an alternative approach, called Continuum GMM
  (henceforth CGMM), that uses the whole continuum of moment conditions associated to the difference between  theoretical and empirical characteristic functions. \newline 
Starting from the general structure of the CGMM approach, we propose a constrained estimation procedure that involves the whole continuum of moment conditions.  Results on asymptotic tail behavior of marginals are used as constraints in order to improve fitting on tails.  An analytical distribution that captures the dependence of extreme events is helpful in many areas such as in portfolio risk management, in reinsurance or in modeling catastrophe risk related to climate change. The proposed estimation procedure is illustrated via numerical analysis on simulated data from a bivariate and trivariate MixedTS distributions. We estimate parameters on bootstrapped samples and investigate their empirical distribution.\newline
The paper is structured as follows. In Section \ref{univ} we give a brief review of the univariate MixedTS, study its asymptotic tail behavior  and discuss the MixedTS L\'evy process. The definition and main features of the multivariate MixedTS distribution are given in Section \ref{multiv}. In Section \ref{est} we explain the estimation procedure  and present  some numerical results. Section \ref{conc} draws some conclusions.\color{black}

\section{Univariate Mixed Tempered Stable}
\label{univ}
Let us recall the definition of a univariate Mixed Tempered Stable
distribution. \begin{definition} A random variable $Y$ is Mixed
Tempered Stable distributed if: 
\begin{equation}
Y= \mu +\beta V+\sqrt{V}X\label{def:Mixedts}
\end{equation}
where parameters $\mu$, $\beta$ $\in\mathbb{R}$ and conditioned on the positive r.v. $V$, $X$ follows
a standardized Classical Tempered Stable distribution with parameters $\left(\alpha,\lambda_{+}\sqrt{V},\lambda_{-}\sqrt{V}\right)$
i.e.: 
\begin{equation}
X|V\sim stdCTS(\alpha,\lambda_{+}\sqrt{V},\lambda_{-}\sqrt{V})\label{cond:univStdTS}
\end{equation} or equivalently 
\begin{equation}
Y|V\sim CTS\left(\alpha,\lambda_+,\lambda_-,
\frac{V}{\Gamma(2-\alpha)(\lambda_+^{\alpha-2}+\lambda_-^{\alpha-2})},
\frac{V}{\Gamma(2-\alpha)(\lambda_+^{\alpha-2}+\lambda_-^{\alpha-2})}, 
\mu + \beta V\right)
\end{equation}where $\alpha \in \left(0,2\right]$ and $\lambda_{+},\lambda_{-}>0$ \citep[see][for more details on CTS]{KIM}. \end{definition}For this distribution it is possible to obtain the first four moments which are reported in the following proposition.
\begin{proposition} The first four moments of the MixedTS have an
analytic expression since: 
\begin{equation}
\left\{ \begin{array}{l}
E\left[Y\right]=\mu+\beta E\left[V\right]\\
Var\left[Y\right]=\beta^{2}Var(V)+E\left[V\right]\\
m_{3}\left(Y\right)=\beta^{3}m_{3}\left(V\right)+3\beta Var(V)+\left(2-\alpha\right)\frac{\left(\lambda_{+}^{\alpha-3}-\lambda_{-}^{\alpha-3}\right)}{\left(\lambda_{+}^{\alpha-2}+\lambda_{-}^{\alpha-2}\right)}E\left[V\right]\\
m_{4}\left(Y\right)=\beta^{4}m_{4}(V)+6\beta^{2}E\left[\left(V-E(V)\right)^{2}V\right]+4\beta\left(2-\alpha\right)\frac{\lambda_{+}^{\alpha-3}-\lambda_{-}^{\alpha-3}}{\lambda_{+}^{\alpha-2}+\lambda_{-}^{\alpha-2}}Var(V)\\
+(3-\alpha)(2-\alpha)\frac{\left(\lambda_{+}^{\alpha-4}+\lambda_{-}^{\alpha-4}\right)}{\left(\lambda_{+}^{\alpha-2}+\lambda_{-}^{\alpha-2}\right)}E\left[V\right].
\end{array}\right.\label{Mom:MixedTS}
\end{equation}
where $m_{3}$ and $m_{4}$ are the third and fourth central moments
respectively. \end{proposition}
We observe that $m_{3}$ and $m_{4}$ depend on the mixing random variable $V$ and on the tempering parameters $\lambda_{-}$ and $\lambda_{+}$. Indeed, we are able to obtain an asymmetric distribution even if we fix $\beta=0$. It is worth to note that parameters $\mu$ and $\beta$ may have an economic interpretation. In particular, $\mu$ can be thought as the risk free rate and $\beta$ as the risk premium of the unit variance process $V$. In the Normal Variance Mean Mixtures is not possible to have negatively skewed distribution with $\beta\ >\ 0$. From an economic point of view, it is not possible to have a positive risk premium for unit variance for negatively skewed distributions. This is a drawback of the Normal Variance Mean Mixture model since negative skewness is frequently observed in financial time series.
\newline The mixture representation becomes very transparent for cumulant
generating functions. Let
\begin{equation}
\Phi_Y(u)=\log E[e^{uY}],\quad
\Phi_V(u)=\log E[e^{uV}],
\end{equation}
and
\begin{equation}
\Phi_H(u)=
\frac{(\lambda_+-u)^\alpha-\lambda_+^\alpha+(\lambda_-+u)^\alpha-\lambda_-^\alpha}{\alpha(\alpha-1)(\lambda_+^{\alpha-2}+\lambda_-^{\alpha-2})}
+
\frac{(\lambda_+^{\alpha-1}-\lambda_-^{\alpha-1})u}{(\alpha-1)(\lambda_+^{\alpha-2}+\lambda_-^{\alpha-2})},
\label{PhiH}
\end{equation}
where $\Phi_H(u)$ is the cumulant generating function of a random variable $H\sim CTS\left(\alpha,\lambda_+,\lambda_-\right)$.
Then we have
\begin{equation}
\Phi_Y(u)=\mu u+\Phi_V(\beta u+\Phi_H(u)).
\end{equation}
As shown in \cite{EditMixedTS2014}, if $V\sim\Gamma(a,b)$, we get some well-known distributions used for modeling financial returns as special cases. For instance if $\alpha=2$ the Variance Gamma introduced in \cite{MadanSeneta1990} is obtained. Fixing $b=\frac{1}{a}$ and letting $a$ go to infinity leads to the Standardized Classical Tempered Stable \cite{KIM}. Choosing:
\begin{eqnarray}
\lambda_{+}&=&\lambda_{-}=\lambda\nonumber\\
a&=&1\nonumber\\ 
b&=&\lambda^{{\alpha-2}}\gamma^{{\alpha}}\left| \frac{\alpha\left(\alpha-1\right)}{\cos\left(\alpha\frac{\pi}{2}\right)}\right|\nonumber\\
\label{conditionForGeoStable}
\end{eqnarray}and computing the limit for $\lambda \rightarrow 0^{+}$ we obtain the Geometric Stable distribution (see \cite{Kozubowski97}).

\subsection{Fundamental strip and moment explosion}
Laplace transform theory tells us that given a random variable $X$ the set of $u\in\mathbb C$ where:
\begin{equation*}
E[|e^{uX}|]<\infty.
\end{equation*} is a strip, which is called the {\em fundamental strip} of $X$. Depending on the tails of $X$ the strip can be the entire set of complex numbers $\mathbb C$, a left or right half-plane, 
a proper strip of finite width or degenerate to the imaginary axis if both tails are heavy. From here on we neglect the case $\alpha=2$ since the MixedTS becomes a Normal Variance Mean Mixture and we refer to \cite{barndorff1982normal} for the fundamental strip and tail behavior in this special case. When $\alpha \in \left(0,2\right)$, $H$  with cumulant generating function in \eqref{PhiH} has fundamental strip $-\lambda_-\leq\Re(u)\leq\lambda_+$.
\begin{theorem}
\label{fundstrip}
Suppose now $V$ has fundamental strip $-\infty<\Re(u)<b$ for some $b>0$. A concrete example would be $V\sim\Gamma(a,b)$.
Then we have:
\begin{enumerate}
\item If $\max\left\{-\beta\lambda_-+\Phi_H(-\lambda_-),\beta\lambda_++\Phi_H(\lambda_+)\right\}<b$ 
then $Y$ has fundamental strip $-\lambda_-\leq\Re(u)\leq\lambda_+$.
\item If $-\beta\lambda_-+\Phi_H(-\lambda_-)<b<\beta\lambda_++\Phi_H(\lambda_+)$ 
then $Y$ has fundamental strip $-\lambda_-\leq\Re(u)<u_+$, where
$u_+$ is the unique real solution to $\beta u+\Phi_H(u)=b$. 

\item If $-\beta\lambda_-+\Phi_H(-\lambda_-)>b>\beta\lambda_++\Phi_H(\lambda_+)$ 
then $Y$ has fundamental strip $u_-\leq\Re(u)<\lambda_+$, where
$u_-$ is the unique real solution to $\beta u +\Phi_H(u)=b$. 

\item If $b<\min\left\{-\beta\lambda_-+\Phi_H(-\lambda_-),\beta\lambda_++\Phi_H(\lambda_+)\right\}$ 
then $Y$ has fundamental strip $u_-<\Re(u)<u_+$ where
$u_-<u_+$ are the two real solutions of $\beta u+\Phi_H(u)=b$.
\end{enumerate}
\end{theorem}

\begin{proof}
First of all we prove point 1 where $Y$ has fundamental strip $-\lambda_-\leq\Re(u)\leq\lambda_+$. Any point $u^{\star}\in \left[-\lambda_-, \lambda_+\right]$ can be written as: 
\[
u^{\star}=\gamma\left(-\lambda_-\right)+\left(1-\gamma\right)\left(\lambda_+\right), \ \ 0\leq \gamma \leq 1.
\]
We start from
\[
\beta u^{\star}+\Phi_H(u^{\star}) = \beta \left[\gamma\left(-\lambda_-\right)+\left(1-\gamma\right)\left(\lambda_+\right)\right]+\Phi_H(\gamma\left(-\lambda_-\right)+\left(1-\gamma\right)\left(\lambda_+\right)).
\]
Since $\Phi_H\left(x\right)$ is a convex function, we have:
\[
\beta u^{\star}+\Phi_H(u^{\star})\leq \beta \left[\gamma\left(-\lambda_-\right)+\left(1-\gamma\right)\left(\lambda_+\right)\right]+\gamma\Phi_H(-\lambda_-)+\left(1-\gamma\right)\Phi_H\left(\lambda_+\right).
\]
Collecting terms with $\gamma$ and $1-\gamma$  we get:
\begin{equation}
\beta u^{\star}+\Phi_H(u^{\star})\leq \gamma \left[\beta \left(-\lambda_-\right)+\Phi_H(-\lambda_-)\right]+\left(1-\gamma\right)\left[\beta \left(\lambda_+\right)+\Phi_H\left(\lambda_+\right)\right].
\label{conustar}
\end{equation}
If $\max\left\{-\beta\lambda_-+\Phi_H(-\lambda_-),\beta\lambda_++\Phi_H(\lambda_+)\right\}<b$  the right hand side in \eqref{conustar} is less than $b$.\newline 
To prove the second point, it is enough to observe that 
\begin{equation*}
G\left(u\right) = \beta u+\Phi_H(u)-b
\end{equation*}is a convex continuous function. Moreover, the condition $-\beta\lambda_-+\Phi_H(-\lambda_-)<b<\beta\lambda_++\Phi_H(\lambda_+)$ implies that:
\[
G\left(-\lambda_-\right)<0, \ \ G\left(\lambda_+\right)>0. 
\]
As observed in \cite{GiaquintaModica2003}, a convex continuous function $f(x)$ in the compact interval $\left[a,b\right]$, $f(a)<0$ and $f(b)>0$ has a unique zero $c \in \left[a,b\right]$. In our case, this result ensures that $u_+$ is the unique real solution of the equation $\beta u+\Phi_H(u)=b.$ Following the same steps in point 1 we get the result in point 2. \newline 
Point 3 is the same of point 2. \newline
Point 4. It is sufficient to observe that $\Phi_H(x)$ is a continuous convex function such that $\Phi_H(0)=0$. For any
 $b >0$, there exists a neighborhood of zero such that $G\left(u\right) < 0$. Continuity and convexity ensure the existence of two zeros $u_-<u_+$. The remaining part of the proof arises from the same steps in point 1.
\end{proof}

\subsection{Tail Behavior of a Mixed Tempered Stable distribution}
\label{tailB}

In order to study the tail behavior of a $MixedTS(\mu,\beta,\alpha,\lambda_+,\lambda_-)-\Gamma(a,b)$, that denotes a MixedTS distribution with Gamma mixing r.v., we need first to recall the structure of its moment generating function where without loss of generality we require $\mu=0$. \newline If $Y \sim MixedTS(0,\beta,\alpha,\lambda_+,\lambda_-)-\Gamma(a,b)$, the moment generating function is defined as:
\begin{equation}
M_{Y}\left(u\right)=E\left(e^{uY}\right)=\left[\frac{b}{b-\left(\beta u + \Phi_H\left(u\right)\right)}\right]^a.
\label{mgfMixedTS}
\end{equation}We recall some useful results on the study of asymptotic tail behavior given in \cite{Friz08}. Given the moment generating function $M$ of a r.v. $X$ with cumulative distribution function $F(x)$ defined as:
\begin{equation*}
M\left(u\right):=\int e^{ux}\mbox{d}F\left(x\right)
\end{equation*}we consider $r^{\star}$ and $q^{\star}$ defined respectively as:
\begin{equation}
-q^{\star}:=\inf\left\{u:M(u)<\infty\right\}, 
\label{qstar}
\end{equation}and
\begin{equation}
r^{\star}:=\sup\left\{u:M(u)<\infty\right\}
\label{rstar}
\end{equation}where $r^{\star},q^{\star} \in (0,\infty)$. Criterion I in \cite{Friz08} for asymptotic study of tails states:
\begin{proposition}
\begin{enumerate}
\item If for some $n\geq 0$, $M^{(n)}(-q^{\star}+u)\sim u^{-\rho} l_1 (1/u)$ for some $\rho >0$, $l_1 \in R_0$ as $u \rightarrow 0^+$ then
\begin{equation*}
\log F((-\infty, -x])\sim -q^{\star}x.
\end{equation*} 
\item If for some $n\geq 0$, $M^{(n)}(r^{\star}-u)\sim u^{-\rho} l_1 (1/u)$ for some $\rho >0$, $l_1 \in R_0$ as $u \rightarrow 0^+$ then
\begin{equation*}
\log F((x, \infty))\sim -r^{\star}x.
\end{equation*}
\end{enumerate}
\label{fri}
\end{proposition} We remark that $R_0$ stands for regularly varying functions of order $0$, i.e. set of slowly varying functions and $M^{(n)}$ the derivative of order $n$ of the moment generating function $M$.

Before studying the tail behavior of the $MixedTS-\Gamma(a,b)$, let us study first the tail behavior of a $CTS\left(\alpha,\lambda_{+},\lambda_{-}\right)$. The fundamental strip is $\left[-\lambda_{-},\lambda_{+}\right]$ and the moment generating function $M_{CTS}$ is:
\begin{equation}
M_{CTS}\left(u\right)=\exp\left[\frac{(\lambda_{+}-u)^{\alpha}-\lambda_{+}^{\alpha}+(\lambda_{-}+u)^{\alpha}-\lambda_{-}^{\alpha}}{\alpha(\alpha-1)(\lambda_{+}^{\alpha-2}+\lambda_{-}^{\alpha-2})}+\frac{(\lambda_{+}^{\alpha-1}-\lambda_{-}^{\alpha-1})u}{(\alpha-1)(\lambda_{+}^{\alpha-2}+\lambda_{-}^{\alpha-2})}\right]
\label{mgfCTS}
\end{equation}
We consider separately two cases: 
\begin{enumerate}
\item $\alpha \in \left(0,1\right)$,
\item $\alpha \in \left[1, 2\right)$. 
\end{enumerate} Considering the right tail of a CTS we have $r^{\star}=\lambda_{+}$ and the $M_{CTS}\left(r^{\star}-s\right)=M_{CTS}\left(\lambda_{+}-s\right)$ converges to constant as $s\rightarrow0^{+}$. \newline
\textbf{CTS case - 1:} Under the assumption that $\alpha\in\left(0,1\right)$, we apply criterion 1 in \cite{Friz08} checking that the first derivative of $M_{CTS}$ satisfies $M_{CTS}^{\left(1\right)}\left(r^{\star}-s\right)=s^{-\rho}l_{1}\left(1/s\right)$ for some $\rho>0$, $l_{1}\in R_{0}$ as $s\rightarrow0^{+}$.
The first derivative of $M_{CTS}$ in \eqref{mgfCTS} is:

\begin{eqnarray*}
M_{CTS}^{\left(1\right)}\left(u\right) & = & M_{CTS}\left(u\right)\left[\frac{(\lambda_{+}-u)^{\alpha-1}-\lambda_{+}^{\alpha-1}-(\lambda_{-}+u)^{\alpha-1}+\lambda_{-}^{\alpha-1}}{(1-\alpha)(\lambda_{+}^{\alpha-2}+\lambda_{-}^{\alpha-2})}\right].
\end{eqnarray*}
Evaluating $M_{CTS}^{\left(1\right)}\left(u\right)$ at point
$\lambda_{+}-s$ and computing the limit for $s\rightarrow0^{+}$, we obtain:
\begin{eqnarray*}
\underset{s\rightarrow0^{+}}{\lim}M_{CTS}^{\left(1\right)}\left(\lambda_{+}-s\right) & \sim & \frac{M_{CTS}\left(\lambda_{+}\right)}{(1-\alpha)(\lambda_{+}^{\alpha-2}+\lambda_{-}^{\alpha-2})}s{}^{-\left(1-\alpha\right)} \ \ \text{for} \ \ \alpha(0,1),
\end{eqnarray*}
where the term $\frac{M_{CTS}\left(\lambda_{+}\right)}{(1-\alpha)(\lambda_{+}^{\alpha-2}+\lambda_{-}^{\alpha-2})}$ is a constant. Therefore we have shown that the first order derivative of $M_{CTS}\left(u\right)$ satisfies criterion 1 in \cite{Friz08} when $\alpha\in\left(0,1\right)$. \newline

\textbf{CTS case - 2} Let us consider now the right tail behavior for $\alpha\in\left[1,2\right)$ where both $M_{CTS}\left(\lambda_{+}-s\right)$ and $M_{CTS}^{\left(1\right)}\left(\lambda_{+}-s\right)$ converge to some constants as $s\rightarrow0^{+}$. We compute the second order derivative of the $M_{CTS}$ and show that criterion 1 in \cite{Friz08} is verified for $n=2$.
\begin{eqnarray*}
M_{CTS}^{\left(2\right)}\left(u\right) & = & M_{CTS}^{\left(1\right)}\left(u\right)\left[\frac{(\lambda_{+}-u)^{\alpha-1}-\lambda_{+}^{\alpha-1}-(\lambda_{-}+u)^{\alpha-1}+\lambda_{-}^{\alpha-1}}{(1-\alpha)(\lambda_{+}^{\alpha-2}+\lambda_{-}^{\alpha-2})}\right]\\
 & + & M_{CTS}\left(u\right)\left[\left(\alpha-1\right)\frac{-(\lambda_{+}-u)^{\alpha-2}-(\lambda_{-}+u)^{\alpha-2}}{(1-\alpha)(\lambda_{+}^{\alpha-2}+\lambda_{-}^{\alpha-2})}\right]\\
 & = & M_{CTS}\left(u\right)\left[\frac{(\lambda_{+}-u)^{\alpha-1}-\lambda_{+}^{\alpha-1}-(\lambda_{-}+u)^{\alpha-1}+\lambda_{-}^{\alpha-1}}{(1-\alpha)(\lambda_{+}^{\alpha-2}+\lambda_{-}^{\alpha-2})}\right]^{2}\\
 & + & M_{CTS}\left(u\right)\left[\frac{(\lambda_{+}-u)^{\alpha-2}+(\lambda_{-}+u)^{\alpha-2}}{(\lambda_{+}^{\alpha-2}+\lambda_{-}^{\alpha-2})}\right]
\end{eqnarray*}
We evaluate $M_{CTS}^{\left(2\right)}\left(u\right)$ at
point $\left(\lambda_{+}-s\right)$ and for $s\rightarrow0^{+}$ we obtain the following result:
\begin{eqnarray*}
\underset{s\rightarrow0^{+}}{\lim}M_{CTS}^{\left(2\right)}\left(\lambda_{+}-s\right) & \sim & M_{CTS}\left(\lambda_{+}\right)\left[\frac{\lambda_{-}^{\alpha-1}-\lambda_{+}^{\alpha-1}-(\lambda_{-}+\lambda_{+})^{\alpha-1}}{(1-\alpha)(\lambda_{+}^{\alpha-2}+\lambda_{-}^{\alpha-2})}\right]^{2}\\
 & + & M_{CTS}\left(\lambda_{+}\right)\left[\frac{(s)^{\alpha-2}+(\lambda_{-}+\lambda_{+})^{\alpha-2}}{(\lambda_{+}^{\alpha-2}+\lambda_{-}^{\alpha-2})}\right]\\
 & \sim & \frac{M_{CTS}\left(\lambda_{+}\right)}{(\lambda_{+}^{\alpha-2}+\lambda_{-}^{\alpha-2})}s^{-\left(2-\alpha\right)}.
\end{eqnarray*}
Now we study the right tail behavior of the $MixedTS-\Gamma(a,b)$. From Theorem \ref{fundstrip} we have that $r^{\star}$ can be $\lambda_+$ or $u_+$, therefore in order to study the behavior of the moment generating function $M_{Y} (u)$ in \eqref{mgfMixedTS} we consider separately two cases:
\begin{itemize}
\item $r^{\star}=\lambda_+$ that refers to points 1 and 3 in Theorem  \ref{fundstrip} where $\Phi_H\left(r^{\star}\right)\neq b-\beta r^{\star}$. 
\item $r^{\star}=u_+$ that refers to points 2 and 4 in Theorem \ref{fundstrip} where $\Phi_H\left(r^{\star}\right)= b-\beta r^{\star}$.
\end{itemize}

\textbf{$MixedTS-\Gamma(a,b)$ case 1:}  $r^{\star}=\lambda_+$ covers case 1 and 3 in Theorem \ref{fundstrip}. The moment generating function of the $MixedTS-\Gamma(a,b)$, defined in \eqref{mgfMixedTS}, at the critical point
$r^{\star}=\lambda_{+}$ is finite. We compute the first order derivative of $M_{Y}\left(u\right)$ and verify if criterion 1 in  \cite{Friz08} is satisfied. We consider separately the two cases $\alpha \in (0,1)$ and $\alpha \in [1,2)$.
\begin{itemize}
\item $r^{\star}=\lambda_+$ and $\alpha \in (0,1)$
\begin{eqnarray*}
M_{Y}^{\left(1\right)}\left(u\right) & = & a\left[\frac{b}{b-\left(\beta u+\Phi_{H}\left(u\right)\right)}\right]^{a-1}\frac{-b}{\left(b-\left(\beta u+\Phi_{H}\left(u\right)\right)\right)^{2}}\left(-\beta-\Phi_{H}^{\prime}\left(u\right)\right)\\
 & = & \left[\frac{b}{b-\left(\beta u+\Phi_{H}\left(u\right)\right)}\right]^{a}\frac{a\left(\beta+\Phi_{H}^{\prime}\left(u\right)\right)}{\left(b-\left(\beta u+\Phi_{H}\left(u\right)\right)\right)}\\
 & = & M_{Y}\left(u\right)\frac{a\left(\beta+\Phi_{H}^{\prime}\left(u\right)\right)}{\left(b-\left(\beta u+\Phi_{H}\left(u\right)\right)\right)}
\end{eqnarray*}
Observing from \eqref{PhiH} that 
\[
\Phi_{H}^{\prime}\left(u\right)=\frac{(\lambda_{+}-u)^{\alpha-1}-\lambda_{+}^{\alpha-1}-(\lambda_{-}+u)^{\alpha-1}+\lambda_{-}^{\alpha-1}}{(1-\alpha)(\lambda_{+}^{\alpha-2}+\lambda_{-}^{\alpha-2})}
\]
and 
\[
\underset{s\rightarrow0^{+}}{\lim}\Phi_{H}^{\prime}\left(\lambda_{+}-s\right)\sim\frac{s{}^{-\left(1-\alpha\right)}}{(1-\alpha)(\lambda_{+}^{\alpha-2}+\lambda_{-}^{\alpha-2})},\ \alpha\in\left(0,1\right)
\]
we obtain:
\begin{eqnarray*}
\underset{s\rightarrow0^{+}}{\lim}M_{Y}^{\left(1\right)}\left(\lambda_{+}-s\right) &=& \underset{s\rightarrow0^{+}}{\lim}M_{Y}\left(\lambda_{+}-s\right)\frac{a\left(\beta+\Phi_{H}^{\prime}\left(\lambda_{+}-s\right)\right)}{\left(b-\left(\beta \left(\lambda_{+}-s\right)+\Phi_{H}\left(\lambda_{+}-s\right)\right)\right)}\\
& \sim &\frac{ M_{Y}\left(\lambda_{+}\right)}{b-\left(\beta\lambda_{+}+\Phi_{H}\left(\lambda_{+}\right)\right)} \frac{a}{(1-\alpha)(\lambda_{+}^{\alpha-2}+\lambda_{-}^{\alpha-2})} s{}^{-\left(1-\alpha\right)} \ \ \text{for} \ \ \alpha \in \left(0,1\right)\\
\end{eqnarray*}

Observe that since the term $\frac{ M_{Y}\left(\lambda_{+}\right)}{b-\left(\beta\lambda_{+}+\Phi_{H}\left(\lambda_{+}\right)\right)} \frac{a}{(1-\alpha)(\lambda_{+}^{\alpha-2}+\lambda_{-}^{\alpha-2})}$ is a positive constant, we can conclude that  the moment generating function of the $MixedTS-\Gamma(a,b)$, in case of $r^{\star}=\lambda_+$ and $\alpha \in \left(0,1\right)$, satisfies criterion 1 in \cite{Friz08} for $n=1$.
\item $r^{\star}=\lambda_+$ and $\alpha \in \left[1,2\right)$. In this particular case both the moment generating function of the $MixedTS-\Gamma(a,b)$ and and its first derivative are constants, therefore we  compute the second derivative of the m.g.f. of the $MixedTS-\Gamma(a,b)$.  
\begin{eqnarray*}
M_{Y}^{\left(2\right)}\left(u\right) & = & \frac{\partial M_{Y}^{\left(1\right)}\left(u\right)}{\partial u}\\
 & = & M_{Y}^{\left(1\right)}\left(u\right)\frac{a\left(\beta+\Phi_{H}^{\prime}\left(u\right)\right)}{\left(b-\left(\beta u+\Phi_{H}\left(u\right)\right)\right)} \\
& + & aM_{Y}\left(u\right)\frac{\left(\Phi_{H}^{\prime\prime}\left(u\right)\right)\left(b-\left(\beta u+\Phi_{H}\left(u\right)\right)\right)+\left(\beta+\Phi_{H}^{\prime}\left(u\right)\right)^{2}}{\left(b-\left(\beta u+\Phi_{H}\left(u\right)\right)\right)^{2}}\\
\end{eqnarray*}
\begin{eqnarray*}
M_{Y}^{\left(2\right)}\left(u\right) & = & M_{Y}\left(u\right)\frac{\left(\beta+\Phi_{H}^{\prime}\left(u\right)\right)^{2}a^{2}}{\left(b-\left(\beta u+\Phi_{H}\left(u\right)\right)\right)^{2}}\\
 & + & aM_{Y}\left(u\right)\frac{\left(\Phi_{H}^{\prime\prime}\left(u\right)\right)\left(b-\left(\beta u+\Phi_{H}\left(u\right)\right)\right)+\left(\beta+\Phi_{H}^{\prime}\left(u\right)\right)^{2}}{\left(b-\left(\beta u+\Phi_{H}\left(u\right)\right)\right)^{2}}\\
 & = & M_{Y}\left(u\right)\frac{\left(\beta+\Phi_{H}^{\prime}\left(u\right)\right)^{2}\left(a^{2}+a\right)}{\left(b-\left(\beta u+\Phi_{H}\left(u\right)\right)\right)^{2}}\\
 & + & \frac{aM_{Y}\left(u\right)}{\left(b-\left(\beta u+\Phi_{H}\left(u\right)\right)\right)}\Phi_{H}^{\prime\prime}\left(u\right).
\end{eqnarray*}
Since $\alpha\in\left[1,2\right)$, we have that the following limit
converges to a positive constant as $s\rightarrow0^{+}$: 
\[
\underset{s\rightarrow0^{+}}{\lim}\Phi_{H}^{\prime}\left(\lambda_{+}-s\right)=\frac{\lambda_{+}^{\alpha-1}+(\lambda_{-}+\lambda_{+})^{\alpha-1}-\lambda_{-}^{\alpha-1}}{(\alpha-1)(\lambda_{+}^{\alpha-2}+\lambda_{-}^{\alpha-2})}>0.
\]
The term $b-\left(\beta\lambda_{+}+\Phi_{H}\left(\lambda_{+}\right)\right)$
is a positive constant term in case 1 and 3 of Theorem \ref{fundstrip} and the same holds for the positive constant term $\frac{aM_{Y}\left(\lambda_{+}\right)}{\left(b-\left(\beta\lambda_{+}+\Phi_{H}\left(\lambda_{+}\right)\right)\right)}$.
Now we study the asymptotic behavior of $\Phi_{H}^{\prime\prime}\left(\lambda_{+}-s\right)$
as $s\rightarrow0^{+}.$
\begin{eqnarray*}
\underset{s\rightarrow0^{+}}{\lim}\Phi_{H}^{\prime\prime}\left(\lambda_{+}-s\right) & = & \underset{s\rightarrow0^{+}}{\lim}\frac{s^{-\left(2-\alpha\right)}+(\lambda_{-}+\lambda_{+})^{\alpha-2}}{(\lambda_{+}^{\alpha-2}+\lambda_{-}^{\alpha-2})}\\
 & \sim & \frac{s{}^{-\left(2-\alpha\right)}}{(\lambda_{+}^{\alpha-2}+\lambda_{-}^{\alpha-2})}
\end{eqnarray*}
Combining these results together we conclude that the moment generating function of the $MixedTS-\Gamma(a,b)$ in case of $r^{\star}=\lambda_+$ and $\alpha \in \left[1,2\right)$ satisfies criterion 1 in \cite{Friz08} for $n=2$ i.e.:
\begin{eqnarray*}
\underset{s\rightarrow0^{+}}{\lim}M_{Y}^{\left(2\right)}\left(\lambda_{+}-s\right) & \sim\frac{aM_{Y}\left(\lambda_{+}\right)}{(\lambda_{+}^{\alpha-2}+\lambda_{-}^{\alpha-2})\left(b-\left(\beta\lambda_{+}+\Phi_{H}\left(\lambda_{+}\right)\right)\right)} & s{}^{-\left(2-\alpha\right)}.
\end{eqnarray*}
\end{itemize}

\textbf{$MixedTS-\Gamma(a,b)$ case - 2} At this point we are left with the case when $r^{\star}=u_{+}$ which covers cases 2 and 4 in Theorem \ref{fundstrip}. We recall that the m.g.f of the $MixedTS-\Gamma(a,b)$ at  point $\left(r^{\star}-s\right)$ is:
\begin{equation}
M^{0}_{Y}\left(r^{\star}-s\right)=\left[\frac{b}{b-\left(\beta \left(r^{\star}-s\right) + \Phi_H\left(r^{\star}-s\right)\right)}\right]^a.
\label{m01}
\end{equation}Multiplying and dividing by $s^a$ in \eqref{m01} we have:
\begin{equation*}
M^0\left(r^{\star}-s\right)=\left[\frac{bs}{b-\left(\beta \left(r^{\star}-s\right) + \Phi_H\left(r^{\star}-s\right)\right)}\right]^a s^{-a}.
\end{equation*}Substituting $s$ with $\frac{1}{t}$ we obtain:
\begin{equation}
M^0\left(r^{\star}-\frac{1}{t}\right)=\left[\frac{b}{tb-\left(\beta \left(tr^{\star}-1\right) + t \Phi_H\left(r^{\star}-\frac{1}{t}\right)\right)}\right]^a t^{a}=g\left(t\right)t^a.
\label{rightail}
\end{equation}
We can show that $g\left(t\right)$ is a slowly varying function that implies $M^0\left(r^{\star}-s\right)$ is a regularly varying function. \newline
Let us study the following limit:
\begin{eqnarray}
\underset{t\rightarrow + \infty}{\lim}\frac{g\left(kt\right)}{g\left(t\right)} & = &
\underset{t\rightarrow + \infty}{\lim}\left[\frac{tb-\left[\beta\left(r^{\star}t-1\right)+t\Phi_H\left(r^{\star}-\frac{1}{t}\right)\right]}{tbk-\left[\beta\left(r^{\star}kt-1\right)+tk\Phi_H\left(r^{\star}-\frac{1}{kt}\right)\right]}\right]^a.
\label{asy1}
\end{eqnarray}
When $r^{\star}=u_+$ we have $\Phi_H\left(r^{\star}\right)= b-\beta r^{\star}$, $g\left(t\right)$ is a slowly varying function, since applying de l'H\^opital theorem  to \eqref{asy1} we get: 
\begin{equation}
\underset{t\rightarrow +\infty}{\lim} \frac{g\left(kt\right)}{g\left(t\right)}=
\left[\underset{t\rightarrow +\infty}{\lim} \frac{\Phi^{\prime}_H\left(r^{\star}-\frac{1}{t}\right)}{\Phi^{\prime}_H\left(r^{\star}-\frac{1}{kt}\right)}\right]^a=1.
\end{equation}
Concluding we can say that in case $r^{\star}=u_+$ the criterion one in \cite{Friz08} is satisfied for $n=0$.

The study of the left tail behavior follows the same steps as above. 
\begin{remark}
In the CTS distribution $\lambda_{-}$ and $\lambda_{+}$ influence both higher moments and tail behavior. The singularities $u^{\star}$ in Theorem \ref{fundstrip} are helpful in describing the asymptotic behavior of the MixedTS tails based on the result:
\begin{equation}
P(Y>y)\sim e^{-u^{\star}y}.
\end{equation}From point 1 in Theorem \ref{fundstrip} we get for the MixedTS the same asymptotic tail behavior as in the CTS, i.e. exponentially decaying, while in the other points of Theorem \ref{fundstrip} the $u^{\star}$'s satisfy the additional condition $\beta u^{\star}+\Phi_H(u^{\star})=b$. Singularities in point 2 and 3  describe respectively right and left asymptotic tail behavior. In point 4 asymptotic of both tails are deduced. The scale parameter $b$ of the mixing r.v. allows us to have more flexibility in capturing tails once skewness and kurtosis, which depend on $\lambda_{-}$ and $\lambda_{+}$, are computed. Consider for example point 2 where $u^{\star}<\lambda_{+}$ that implies, for fixed $y$, $e^{-\lambda_{+}y}<e^{-u^{\star}y}$ from where we deduce that a higher weight is given to the right tail of the MixedTS than in the CTS case.
\end{remark}

We conclude this section by investigating numerically the implications of 
Proposition \ref{fri}. Results on the behavior of tails can be used for the identification of
\(q^{\star}\) and \(r^{\star}\) in \eqref{qstar} and in \eqref{rstar}. Indeed,
for $x \rightarrow -\infty$, we have:
\begin{equation}
 \log\left[F\left(x\right)\right] = -q^{\star}x + o\left(x\right),
\label{lef1} 
 \end{equation}while, for $x \rightarrow +\infty$, we obtain:
\begin{equation}
 \log\left[1-F\left(x\right)\right] = -r^{\star}x + o\left(x\right). 
\label{rig1}
 \end{equation}Figure \ref{tailbeh} refers to the behavior of
\(\log\left[F\left(x\right)\right]\) and of
\(\log\left[1-F\left(x\right)\right]\) for the
MixedTS-$\Gamma\left(1,1\right)$ with parameters \(\mu = 0\), \(\beta =0\), \(\alpha= 1.25\), \(\lambda_{+} = 1.2\) and \(\lambda_{-} =1.9 \).
\begin{figure}[!h]
\begin{center}
		{\includegraphics[trim = 2cm 0mm 3cm 0mm, width=0.5\textwidth]{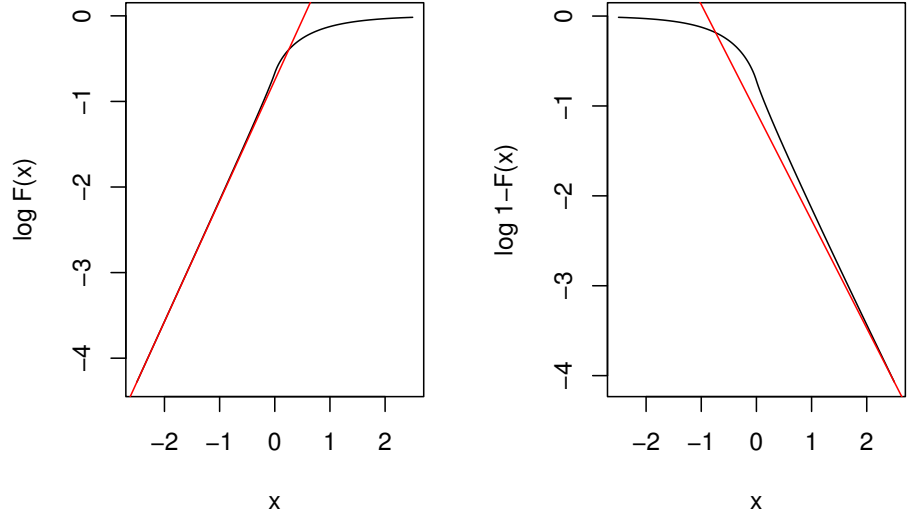}}
	\caption{Left and right tail behavior of the MixedTS-$\Gamma\left(1,1\right)$ with fixed parameters  \(\mu = 0\), \(\beta =0\),  \(\alpha= 1.25\), \(\lambda_{+}= 1.2\) and \(\lambda_{-}= 1.9\).	\label{tailbeh}}	
	\end{center}
\end{figure}Considering relations in \eqref{lef1} and in \eqref{rig1}, we estimate \(q^{\star}\) and \(r^{\star}\) as
the slope of two linear regressions following four steps: $(i)$ Given a sample composed by $\tilde{n}$ observations, we determine the empirical cumulative distribution function $\widehat{F}_{\tilde{n}}\left(x\right)=\frac{1}{\tilde{n}}\sum_{i=1}^{\tilde{n}}\mathbf{1}_{x_i\leq x}$. $(ii)$ Then we determine $\hat{x}_{\zeta}$ and $\hat{x}_{1-\zeta}$ as the empirical quantiles at level $\zeta$ and $1-\zeta$, i.e.: 
\[
\hat{x}_{\zeta}:=\inf\left\{x_{i}: \widehat{F}_{\tilde{n}}\left(x_{i}\right)\geq \zeta\right\}
\] 
and
\[
\hat{x}_{1-\zeta}:=\inf\left\{x_{i}: \widehat{F}_{\tilde{n}}\left(x_{i}\right)\geq 1-\zeta\right\}.
\]The set [$x_{(1)}$, $x_{(2)}$, $\ldots$ ,$x_{(\tilde{n})}$] refers to the sorted values from the smallest $x_{(1)}$ to the largest $x_{(\tilde{n})}$. \newline $(iii)$ We introduce the sets $\mathcal{L}_{\tilde{n}}\left(\zeta\right)$ and $\mathcal{U}_{\tilde{n}}\left(\zeta\right)$ defined as:
\begin{eqnarray}
\mathcal{L}_{\tilde{n}}\left(\zeta\right):=\left\{\left(x_{i},\widehat{F}_{\tilde{n}}\left(x_{i}\right)\right): x_{i} \in \left[x_{(1)}, \hat{x}_{\zeta}\right]\right\}\nonumber\\
\mathcal{U}_{\tilde{n}}\left(\zeta\right):=\left\{\left(x_{i},\widehat{F}_{\tilde{n}}\left(x_{i}\right)\right): x_{i} \in \left[\hat{x}_{1-\zeta}, x_{(\tilde{n})}\right]\right\}.\nonumber\\
\end{eqnarray} $(iv)$ We use the elements in the set $\mathcal{L}_{\tilde{n}}$ to estimate $q^{\star}$ as the slope of the linear regression:
\[
\log[\widehat{F}_{\tilde{n}}(x_i)]=-q^{\star} x_i+\epsilon_i, \ \ \left(x_{i},\widehat{F}_{\tilde{n}}(x_i)\right) \in \mathcal{L}_n\left(\zeta\right).
\]while the elements in the set $\mathcal{U}_{\tilde{n}}$ are used for the estimation of the coefficient $r^{\star}$ as the slope of the following regression:
\[
\log[1-\widehat{F}_{\tilde{n}}(x_i)]=-r^{\star} x_i+\epsilon_i, \ \ \left(x_{i},\widehat{F}_{\tilde{n}}(x_i)\right) \in \mathcal{U}_n\left(\zeta\right),
\] where $\epsilon_i$ is an error term. In Figure \ref{varalpha} we show the behavior of the estimated
\(q^{\star} \) and \(r^{\star} \) for varying
\(\zeta\) if true values are $q^{\star}=1.4105$ and $r^{\star}=1.2$. This result is useful in estimation of a MixedTS-$\Gamma$ since it can be used as a constraint in the optimization routine when we require the empirical
\(q^{\star}\) and \(r^{\star}\) to be equal to the corresponding counterpart.
\begin{figure}[!h]
\begin{center}
		{\includegraphics[trim = 2cm 0mm 3cm 0mm, width=0.5\textwidth]{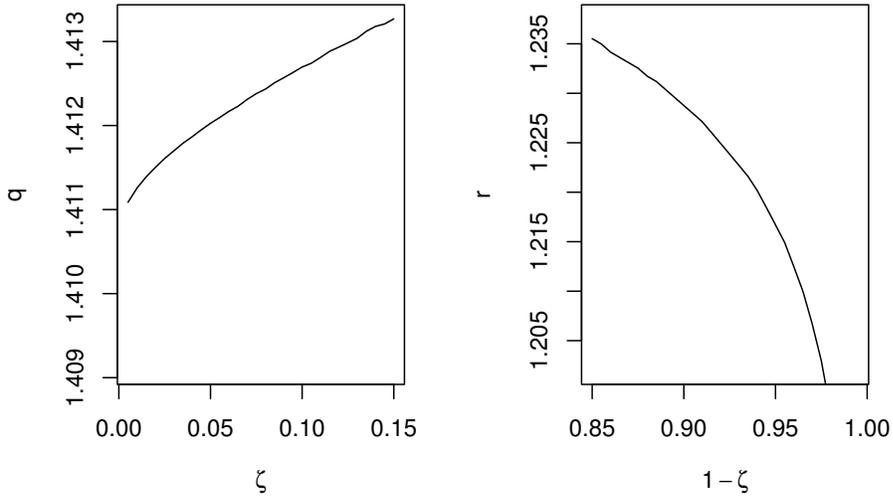}}
	\caption{Behavior of the estimated \(q^{\star}\) and \(r^{\star}\) for different levels of $\zeta$.\label{varalpha}}	
	\end{center}
\end{figure}
\subsection{MixedTS L\'evy process}
Suppose $F$ is an infinitely divisible distribution on $\mathbb R_+$ with cumulant function $\Phi_V$.
Then there is a convolution semigroup of probability measures $(F_t)_{t\geq0}$ on $\mathbb R_+$
and a L\'evy process $(V_t)_{t\geq0}$ such that $V_t\sim F_t$ for $t\geq0$ and $\Phi_{V_t}(u)=t\Phi_V(u)$.\newline
In \cite{EditMixedTS2014} it is shown that the $MixedTS\left(\mu t, \beta,  \alpha, \lambda_{+}, \lambda_{-}\right)- F_t$ distribution is infinitely divisible. According to the general theory, see for example Prop.3.1, p.69 in \cite{ContTankov2004}, there exists a L\'evy process $(Y_t)_{t\geq0}$ such that $Y_1\sim\ MixedTS\left(\mu, \beta, \alpha, \lambda_{+}, \lambda_{-}\right)-F_1$. We have 
\begin{equation}
\Phi_{Y_t}(u)=\mu tu+\Phi_{V_t}(\beta u+\Phi_H(u)),
\end{equation}thus if $Y_1$ is $MixedTS(\mu,\beta,\alpha,\lambda_+,\lambda_-)$ with mixing distribution $F_1$,
then $Y_t$ is $MixedTS(\mu t,\beta,\alpha,\lambda_+,\lambda_-)$ with mixing distribution $F_t$. In the case $Y_1$ is $MixedTS(\mu,\beta,\alpha,\lambda_+,\lambda_-)$ with mixing distribution $\Gamma(a,b)$
then $Y_t$ is $MixedTS(\mu t,\beta,\alpha,\lambda_+,\lambda_-)$ with mixing distribution $\Gamma(at,b)$, since $V_1\sim \Gamma(a,b)$ implies $V_t\sim \Gamma(at,b)$.
\begin{definition}
A L\'evy process $(Y_t)_{t\geq0}$ such that
$Y_1\sim MixedTS\left(\mu, \beta, \alpha, \lambda_{+}, \lambda_{-}\right)-F_1$ is called the  {\em $MixedTS\left(\mu t, \beta, \alpha, \lambda_{+}, \lambda_{-}\right)-F_t$ L\'evy process}.
\end{definition}The $MixedTS\left(\mu t, \beta, \alpha, \lambda_{+}, \lambda_{-}\right)-F_t$ L\'evy process is first of all a L\'evy process, thus it starts at zero and has independent and stationary increments,
and we have for $0\leq s<t$
\begin{equation}
Y_t-Y_s\sim MixedTS(\mu(t-s),\beta,\alpha,\lambda_+,\lambda_-)-F_{t-s}.
\end{equation}For example with gamma mixing
\begin{equation}
Y_t-Y_s\sim MixedTS(\mu(t-s),\beta,\alpha,\lambda_+,\lambda_-)-\Gamma(a(t-s),b).
\end{equation}
We conclude this section by showing how to determine the $MixedTS$ L\'evy measure from a numerical point of view.\newline
The L\'evy-Khintchine formula says
\begin{equation}
\Phi_Y(u)= i\mu u+\int_{\mathbb{R}-0}\left(e^{iux}-1-iux\mathds{1}_{\left|x\right|\leq 1}\right)g_{Y}\left(x\right)dx,
\end{equation} where $\mathds{1}$ is the indicator function, $g_Y$ is the $MixedTS$ L\'evy density. Differentiating twice yields:
\begin{equation}
\Phi_Y''(u)=\int_{-\infty}^{\infty}-e^{iux}x^2g_Y(x)dx.
\label{aa}
\end{equation}
Choosing $x=-z$, the integral in \eqref{aa} becomes:
\begin{equation}
\Phi_Y''\left(u\right)=\int_{-\infty}^{\infty}e^{-iuz}z^2g_Y\left(-z\right)dz.
\label{aab}
\end{equation}
Therefore $\Phi_Y''\left(u\right)$ is the bilateral transform of $z^2g_Y\left(-z\right)$ with $z=-x$. The L\'evy density $g_Y\left(x\right)$ is determined using the Bromwhich inversion integral 
\citep[see][for details]{Boas:913305}. In Figure \ref{fig:mtslevy}, we have the L\'evy density of the $MixedTS-\Gamma\left(1,1\right)$ with fixed parameters $\mu = 0$, $\beta=0$, $\alpha=1.25$, $\lambda_+=1.9$ and $\lambda_-=1.9$
\begin{figure}[h!]
	\centering
		\includegraphics[trim = 2cm 0mm 3cm 0mm, width=.40\textwidth]{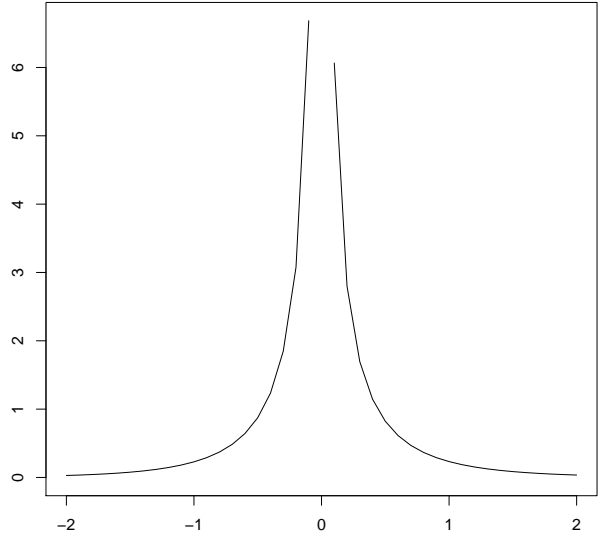}
	\caption{L\'evy density of the $MixedtTS-\Gamma\left(1,1\right)$ with fixed parameters $\mu = 0$, $\beta=0$, $\alpha=1.25$, $\lambda_+=1.9$ and $\lambda_-=1.9$.\label{fig:mtslevy}}
\end{figure}


\section{Multivariate Mixed Tempered Stable}
\label{multiv}
In this section we define the multivariate MixedTS distribution, analyze its characteristics in the particular case the mixing r.v. is  multivariate Gamma distributed.

\subsection{Definition and properties}
\begin{definition}
\label{multivMixeddef}
A random vector $Y\in \mathbb{R^N}$ follows
a multivariate MixedTS distribution if the $i^{th}$ component has the following form: 
\begin{equation}
Y_{i}=\mu_{i}+\beta_{i}V_{i}+\sqrt{V_{i}}X_{i},\label{comp:mixedts}
\end{equation}
where $V_{i}$ is the $i^{th}$ component of a random
vector $V$, defined as: 
\begin{equation}
V_{i}=G_{i}+a_{i}\ Z,
\label{Vh}
\end{equation}$G_{i}$ and $Z$ are infinitely indivisible defined on $\mathbb R_+$ with
$\left\{ G_{i}\right\} _{i=1}^{N}$ and $Z$ mutually independent; $a_{i} \geq 0$ and
\begin{equation}
X_{i}|V_{i}\sim stdCTS\left(\alpha_{i},\lambda_{+,i}\sqrt{V_{i}},\lambda_{-,i}\sqrt{V_{i}}\right).
\label{condXV}
\end{equation}
\end{definition}
It is worth to notice that in Definition \ref{multivMixeddef} it is possible to consider a finer sigma field 
$\chi=\sigma\left(\left\{ G_{i}\right\} _{i=1,\ldots,N},\ Z\right)$
 generated from the sequence of r.v.'s $\left\{ G_{i}\right\} _{i=1,\ldots,N},\ Z$. Let us define $\left\{ H_{i}\right\} _{i=1}^{N}$ as: 
\begin{equation}
H_{i}:=X_{i}|\chi\ \text{for }i=1,\ldots,N,\label{cond:multiMixedTS}
\end{equation}and require the distribution of $H_{i}$ to be a Standardized Classical
Tempered Stable: 
\begin{equation}
H_{i}\sim stdCTS\left(\alpha_{i},\ \lambda_{+,i}\sqrt{V_{i}},\ \lambda_{-,i}\sqrt{V_{i}}\right).\label{cond:Hi}
\end{equation}
Notice that this condition is the generalization of \eqref{condXV}
since the following implications hold: 
\footnotesize{\begin{eqnarray*}
H_{i}\sim stdCTS\left(\alpha_{i},\lambda_{+,i}\sqrt{V_{i}},\lambda_{-,i}\sqrt{V_{i}}\right) & \Rightarrow& X_{i}|\left\{ G_{i},Z\right\} \sim stdCTS\left(\alpha_{i},\lambda_{+,i}\sqrt{V_{i}},\lambda_{-,i}\sqrt{V_{i}}\right)\\
X_{i}|\left\{ G_{i},Z\right\} \sim stdCTS\left(\alpha_{i},\lambda_{+,i}\sqrt{V_{i}},\lambda_{-,i}\sqrt{V_{i}}\right)&\Rightarrow&  X_{i}|V_{i}\sim stdCTS\left(\alpha_{i},\lambda_{+,i}\sqrt{V_{i}},\lambda_{-,i}\sqrt{V_{i}}\right).
\end{eqnarray*}}
\normalsize
The sigma field $\chi$ is also suitable in order to define the dependence
structure between components since we impose independence among $H_{i}$'s.\newline
We remark that if  $G_{i}\sim\Gamma(l_{i},m_{i})$, $Z\sim\Gamma(n,k)$ and for each $i=1,...,N$ :
\[
a_{i}=\frac{k}{m_{i}}\ \ \ \rightarrow\ a_{i}Z\sim\ \Gamma(n,m_{i})
\]we have that $V_{i}$ is sum of two Gamma's with the same scale parameter.
Applying the summation property, we have $V_{i}\sim\ \Gamma(l_{i}+n,m_{i})$ that guarantees infinite divisibility, necessary for definition of multivariate MixedTS-$\Gamma$.
\begin{remark}
The multivariate MixedTS definition in \eqref{comp:mixedts} using matrix notation reads:
\begin{equation}
\mathbf{Y}=\mathbf{\mu}+B\mathbf{V}+S^{\frac{1}{2}}\mathbf{X}
\end{equation}where $\mathbf{\mu} \in \mathbb{R^N}$, $B \in \mathbb{R^{N\times N}}$ such that $B=diag\left(\beta_{1}, \ldots \beta_{N}\right)$, $\mathbf{V} \in\mathbb{R^N}$ is a random vector with positive elements, $S$ is a random matrix positive defined, such that $S=diag\left(V_{1}, \ldots V_{N}\right)$ and $X$ is a standardized Classical Tempered Stable random vector.
\end{remark}

The characteristic function of the multivariate MixedTS has a closed form formula as reported in the following proposition (for the derivation see Appendix \ref{dmcf}).
\begin{proposition} The characteristic function of the multivariate
MixedTS is: 
\begin{equation}
\begin{tabular}{l}
 \ensuremath{\varphi_{Y}(u)=E\left[\exp\left(iuY\right)\right]} \\
 \ensuremath{\ \ \ \ \ \ \ \ \ \ \ \ \ \ =e^{i\sum\limits _{h=1}^{N}u_{h}\mu_{h}}e^{\Phi_{Z}\left(\sum\limits _{h=1}^{N}\left(i\ a_{h}u_{h}\beta_{h}+a_{h}L_{stdCTS}\left(u_h;\lambda_{+,h},\lambda_{-,h},\alpha_{h}\right)\right)\right)}} \\
 \ensuremath{\ \ \ \ \ \ \ \ \ \ \ \ \ \ \ast\prod\limits _{h=1}^{N}e^{\Phi_{G_{h}}\left(i\ u_{h}\beta_{h}+L_{stdCTS}\left(u_h;\lambda_{+,h},\lambda_{-,h},\alpha_{h}\right)\right)},} 
\end{tabular}\label{charMixedTS}
\end{equation}where the $L_{stdCTS}\left(u;\alpha,\lambda_{+},\lambda_{-}\right)$
is the characteristic exponent of a standardized Classical Tempered Stable r.v.
defined as: \\
$L_{stdCTS}\left(u;\ \lambda_{+},\ \lambda_{-}, \ \alpha \right)=\frac{\left(\lambda_{+}-iu\right)^{\alpha}-\lambda_{+}^{\alpha}+\left(\lambda_{-}+iu\right)^{\alpha}-\lambda_{-}^{\alpha}}{\alpha\left(\alpha-1\right)\left(\lambda_{+}^{\alpha-2}+\lambda_{-}^{\alpha-2}\right)}\ +\frac{iu\left(\lambda_{+}^{\alpha-1}-\lambda_{-}^{\alpha-1}\right)}{\left(\alpha-1\right)\left(\lambda_{+}^{\alpha-2}+\lambda_{-}^{\alpha-2}\right)}.$\end{proposition}
\begin{proposition} 
Consider a random vector $\mathbf{Y}$ where the distribution of each component is $Y_{i}\sim MixedTS-\Gamma\left(l_i+n,m_i\right)$ for $i=1,\ldots,N$. The formulas for the moments are:
\begin{itemize}
\item Mean of the general $i^{th}$ element: 
\begin{equation}
E\left[Y_{i}\right]=\mu_{i}+\beta_{i}\frac{l_{i}+n}{m_{i}}.\label{meani}
\end{equation}
\item Variance $\sigma_{i}^{2}$
of the $i^{th}$ element: 
\begin{equation}
\sigma_{i}^{2}=\left(1+\frac{\beta_{i}^{2}}{m_{i}}\right)\frac{\left(l_{i}+n\right)}{m_{i}}.\label{vari}
\end{equation}
\item Covariance $\sigma_{ij}$ between the $i^{th}$ and $j^{th}$ elements: 
\begin{equation}
\sigma_{ij}=\frac{\beta_{i}\beta_{j}}{m_{i}m_{j}}n.\label{covar}
\end{equation}
\item Third central moment of the $i^{th}$ component: 
\begin{align}
m_{3} & =\left[\left(2-\alpha_{i}\right)\frac{\lambda_{+,i}^{\alpha_{i}-3}-\lambda_{-,i}^{\alpha_{i}-3}}{\lambda_{+,i}^{\alpha_{i}-2}+\lambda_{-,i}^{\alpha_{i}-2}}+\left(3+2\frac{\beta_{i}^{2}}{m_{i}}\right)\frac{\beta_{i}}{m_{i}}\right]\frac{\left(l_{i}+n\right)}{m_{i}}.\label{skew}
\end{align}
\item Fourth central moment of the $i^{th}$ element:
\begin{align}
m_{4} & =\beta_{i}^{4}\left(3+\frac{6}{l_{i}+n}\right)\frac{\left(l_{i}+n\right)^{2}}{m_{i}^{4}}+6\beta_{i}^{2}\frac{l_{i}+n}{m_{i}^{3}}\left(l_{i}+n+2\right)+\nonumber \\
 & +4\beta_{i}\left(2-\alpha_{i}\right)\left(\frac{\lambda_{+,i}^{\alpha_{i}-3}-\lambda_{-}^{\alpha_{i}-3}}{\lambda_{+,i}^{\alpha_{i}-2}+\lambda_{-,i}^{\alpha_{i}-2}}\right)\frac{l_{i}+n}{m_{i}^{2}}+\left(3-\alpha_{i}\right)\left(2-\alpha_{i}\right)\left(\frac{\lambda_{+,i}^{\alpha_{i}-4}+\lambda_{-,i}^{\alpha_{i}-4}}{\lambda_{+,i}^{\alpha_{i}-2}+\lambda_{-,i}^{\alpha_{i}-2}}\right)\frac{l_{i}+n}{m_{i}}\label{kurt}.
\end{align}
\end{itemize}
\end{proposition}
See Appendix \ref{Ap1} for details on moment derivation. 
From \eqref{covar} and \eqref{skew} is evident that the multivariate $MixedTS-\Gamma$ overcomes the limits of the multivariate Variance Gamma distribution in capturing the dependence structure between components (see \cite{hitaj2013hedge}). Indeed, the relation that exists between the sign of the skewness of two marginals and the sign of their covariance in the multivariate Variance Gamma, is broken up by the tempering parameters in the multivariate $MixedTS-\Gamma$.   

In particular the following result determines the existence of upper and lower
bounds for the covariance depending on the tempering parameters. Here we consider the cases that the Semeraro model is not able to capture.

\begin{theorem}
Let $Y_{i}$ and $Y_{j}$ be two components of a multivariate MixedTS-$\Gamma$,
the following results hold:

\begin{itemize}
\item[1] $\underline{\sigma}_{ij}:=\frac{\beta_{i}^{\ast}\beta_{j}^{\ast}%
}{m_{i}m_{j}}n\leq\sigma_{ij}$ where $\sigma_{ij}$ is defined in
\eqref{vari} and $\underline{\sigma}_{ij}<0$ if $skew\left(  Y_{i}\right)
\geq0$, $skew\left(  Y_{j}\right)  \geq0$ and $\lambda_{+,i}^{{}}%
\gtrless\lambda_{-,i}^{{}}\ \wedge\ \lambda_{+,j}^{{}}\lessgtr\lambda
_{-,j}^{{}}.$

\item[2] $\overline{\sigma}_{ij}:=\frac{\beta_{i}^{\ast}\beta_{j}^{\ast}%
}{m_{i}m_{j}}n\geq\sigma_{ij}$ and $\overline{\sigma}_{ij}<0$ if $skew\left(
Y_{i}\right)  \leq0,\ skew\left(  Y_{j}\right)  \leq0$ and $\lambda_{+,i}^{{}%
}\gtrless\lambda_{-,i}^{{}}\ \wedge\ \lambda_{+,j}^{{}}\lessgtr\lambda
_{-,j}^{{}}.$

\item[3] $\underline{\sigma}_{ij}=-\infty$ \ and $\overline{\sigma}_{ij}
=+\infty$ if $skew\left(  Y_{i}\right)  \leq0,\ skew\left(  Y_{j}\right)
\geq0$ or $skew\left(  Y_{i}\right)  \geq0,\ skew\left(  Y_{j}\right)  \leq0.$
\end{itemize}

\begin{proof}
Let us first discuss the case where both components have positive skewness. In this
case the lower bound of the covariance exists if the following problem admits a
solution:%
\begin{equation}
\begin{array}
[c]{c}%
\underline{\sigma}_{ij}:=\underset{\beta_{i}\beta_{j}}{\min}\frac{\beta
_{i}\beta_{j}}{m_{i}m_{j}}n\\
skew\left(  Y_{i}\right)  \geq0\\
skew\left(  Y_{j}\right)  \geq0
\end{array}
.\label{ProbLowCovSkBothPos}%
\end{equation}
The signs of skewness depend on the signs of the following quantities:%
\[%
\begin{array}
[c]{c}%
(2-\alpha_{i})\left(  \frac{\lambda_{+,i}^{\alpha_{i}-3}-\lambda_{-,i}%
^{\alpha_{i}-3}}{\lambda_{+,i}^{\alpha_{i}-2}+\lambda_{-,i}^{\alpha_{i}-2}%
}\right)  +3\frac{\beta_{i}}{m_{i}}+2\frac{\beta_{i}^{3}}{m_{i}^{2}}\geq0\\
(2-\alpha_{j})\left(  \frac{\lambda_{+,j}^{\alpha_{j}-3}-\lambda_{-,i}%
^{\alpha_{j}-3}}{\lambda_{+,j}^{\alpha_{j}-2}+\lambda_{-,j}^{\alpha_{j}-2}%
}\right)  +3\frac{\beta_{j}}{m_{j}}+2\frac{\beta_{j}^{3}}{m_{j}^{2}}\geq0
\end{array}
.
\]
The feasible region $S_{a}$ of the minimization problem in
\eqref{ProbLowCovSkBothPos} depends on the difference between tempering
parameters. We observe that the cubic function $g\left(  \beta_{i}\right)
:=(2-\alpha_{i})\left(  \frac{\lambda_{+,i}^{\alpha_{i}-3}-\lambda
_{-,i}^{\alpha_{i}-3}}{\lambda_{+,i}^{\alpha_{i}-2}+\lambda_{-,i}^{\alpha
_{i}-2}}\right)  +3\frac{\beta_{i}}{m_{i}}+2\frac{\beta_{i}^{3}}{m_{i}^{2}}$
is strictly increasing and satisfies the following limits:%
\begin{align*}
\underset{\beta_{i}\rightarrow+\infty}{\lim}g\left(  \beta_{i}\right)   &
=+\infty\\
\underset{\beta_{i}\rightarrow-\infty}{\lim}g\left(  \beta_{i}\right)   &
=-\infty.
\end{align*}
Therefore exists only one $\beta_{i}^{\ast}$ such that $g\left(  \beta
_{i}^{\ast}\right)  =0.$ The sign of $\beta_{i}^{\ast}$ is determined by the
following implications:%
\begin{align*}
\lambda_{+,i}^{{}}  &  =\lambda_{-,i}^{{}}\Longrightarrow g\left(  0\right)
=0\Longrightarrow\beta_{i}^{\ast}=0\\
\lambda_{+,i}^{{}}  &  >\lambda_{-,i}^{{}}\Longrightarrow g\left(  0\right)
<0\Longrightarrow\beta_{i}^{\ast}>0\\
\lambda_{+,i}^{{}}  &  <\lambda_{-,i}^{{}}\Longrightarrow g\left(  0\right)
>0\Longrightarrow\beta_{i}^{\ast}<0.
\end{align*}
The feasible region can be written as:%
\[
S_{a}=\left\{  \left(  \beta_{i},\beta_{j}\right)  :\beta_{i}\geq\beta
_{i}^{\ast}\wedge\beta_{j}\geq\beta_{j}^{\ast}\right\}
\]
and the lower bound is $\underline{\sigma_{ij}}=\frac{\beta_{i}^{\ast}\beta
_{j}^{\ast}}{m_{i}m_{j}}n$ while the upper bound is $\overline{\sigma_{ij}%
}=+\infty.$In this case the lower bound is negative when%
\[
\lambda_{+,i}^{{}}>\lambda_{-,i}^{{}}\ \wedge\ \lambda_{+,j}^{{}}%
<\lambda_{-,j}^{{}}%
\]
or
\[
\lambda_{+,j}^{{}}>\lambda_{-,j}^{{}}\ \wedge\ \lambda_{+,i}^{{}}%
<\lambda_{-,i}^{{}}.%
\]
Now we consider the case when both skewnesses are negative. Following a similar
procedure the feasible region becomes:%
\[
S_{a}=\left\{  \left(  \beta_{i},\beta_{j}\right)  :\beta_{i}\leq\beta
_{i}^{\ast}\wedge\beta_{j}\leq\beta_{j}^{\ast}\right\}
\]
The $\underline{\sigma_{ij}}=-\infty$ and the upper bound is $\overline
{\sigma_{ij}}=\frac{\beta_{i}^{\ast}\beta_{j}^{\ast}}{m_{i}m_{j}}n$. The upper
bound is negative when
\[
\lambda_{+,i}^{{}}>\lambda_{-,i}^{{}}\ \wedge\ \lambda_{+,j}^{{}}%
<\lambda_{-,j}^{{}}%
\]
or
\[
\lambda_{+,j}^{{}}>\lambda_{-,j}^{{}}\ \wedge\ \lambda_{+,i}^{{}}%
<\lambda_{-,i}^{{}}.
\]
The last case refers to the context when the skewnesses have different signs and following the
same procedure as above we have $\underline{\sigma_{ij}}=-\infty$ \ and
$\overline{\sigma_{ij}}=+\infty.$
\end{proof}
\end{theorem}

\subsection{Simulation scheme}
The structure of the univariate / multivariate MixedTS distribution allows us to exploit the procedures (algorithms) for the estimation of the Tempered Stable proposed in literature for instance in \cite{KIM}. 
The steps that we follow for the simulation of a multivariate MixedTS with $N$ components are listed below.
\begin{enumerate}
\item \texttt{Simulate independent random variables $G_i\sim \Gamma\left(l_i,m_i\right)$ and $Z\sim \Gamma\left(n,k\right)$ for $i= 1\ \ldots\ N$.}
\item \texttt{Compute $V_i = G_i + a_i Z$ for $i= 1\ \ldots\ N$.}
\item \texttt{Simulate $X_i |V_i\ \sim\ stdCTS\left(\alpha_i,\ \lambda_{+,i} \sqrt{V_i},\ \lambda_{-,i} \sqrt{V_i} \right)$.}
\item \texttt{Compute $Y_i=\mu_{i}+\beta_{i} V_{i}+\sqrt{V_i}X_i$.}
\item \texttt{Repeat the steps from $1$ to $4$.}
\end{enumerate}
The multivariate MixedTS inherits from its univariate version a similar level of flexibility. For instance, choosing all $\alpha_i=2$ for $i=1,..,N$ we obtain the multivariate Variance Gamma introduced in \cite{semeraro2008multivariate} as a special case. As observed in \cite{hitaj2013hedge}, the Semeraro's model is not able to capture some situations often observed in financial time series. We recall that Semeraro's model has the same structure as in \eqref{comp:mixedts} but instead of each $X_{i}$ we have $W_{i}$ where $W_{1},..,W_{N}$ are independent Standard Normals. This structure limits the capacity of the multivariate Variance Gamma distribution in capturing different dependence structures between components of a random vector, as the sign of skewness is determined by the sign of $\beta$ and the covariance between components has the same form as in \eqref{covar}. In particular this distribution is not able to reproduce negatively correlated components with marginal negative (or positive) skewness or positively correlated components with different signs on the marginal skewness. The multivariate $MixedTS-\Gamma$ overcomes these limits as the sign of marginal skewness depends on $\beta$ and on the tempering parameters. \newline
In Figures \ref{corr_dep} and \ref{corr_dep2} we report the level curves of joint densities of bivariate $MixedTS-\Gamma$ and the corresponding marginal densities. In the Figure \ref{corr_dep} we consider the case where the marginal distributions have opposed signs for skewness ($skew(Y_{1})=7.37$ and $skew(Y_{2})=-19.11$) and positive correlation. In the Figure \ref{corr_dep2} the components are negatively correlated with marginal negative skew distributions ($skew(Y_{1})=-3$ and $skew(Y_{2})=-19.3$). These cases can not be reproduced using the Semeraro model.

 

\begin{figure}[!h]
\begin{center}
		{\includegraphics[trim = 2cm 0mm 3cm 0mm, width=0.7\textwidth]{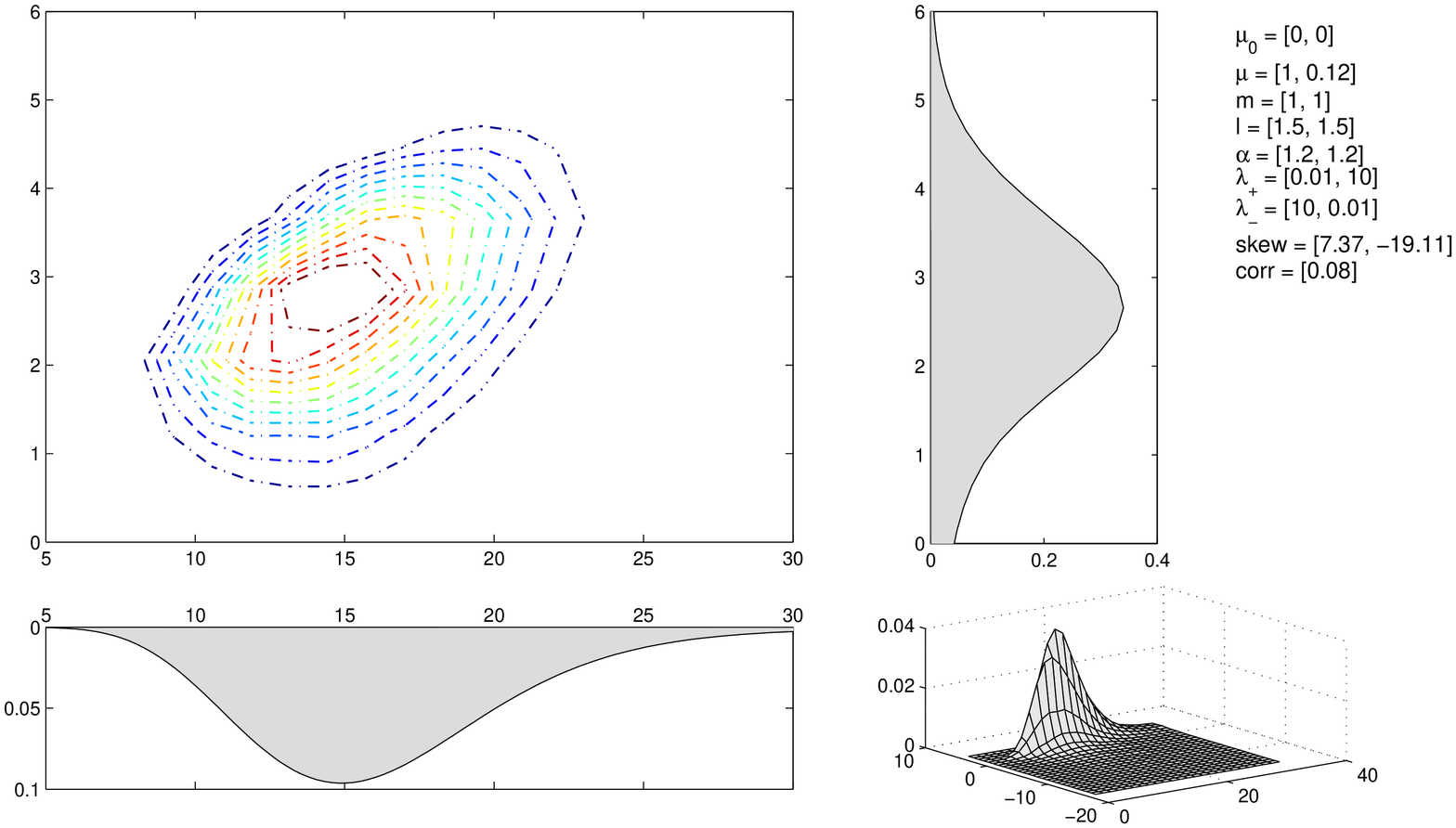}}
	\caption{Level curves, marginal distributions and joint density for a bivariate $MixedTS-\Gamma$ with $n=15$.	\label{corr_dep}}	
	\end{center}
\end{figure}

\begin{figure}[!h]
\begin{center}
	{\includegraphics[trim = 2cm 0mm 3cm 0mm, width=0.7\textwidth]{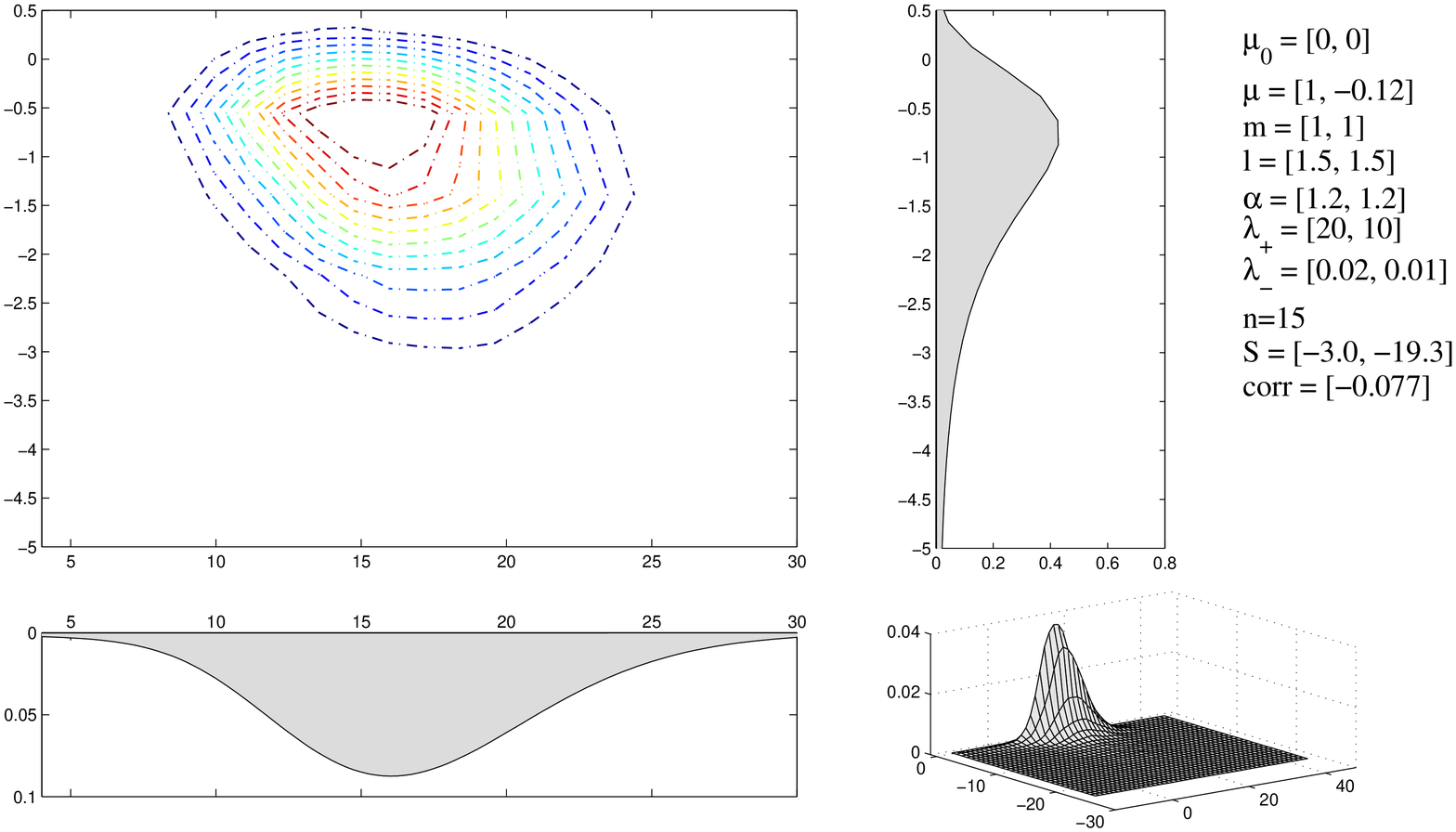}}
	\caption{Level curves, marginal distributions and joint density for a bivariate $MixedTS-\Gamma$ with $n=15$. 	\label{corr_dep2}}
	\end{center}
\end{figure}

\section{Estimation procedure}
\label{est}
In this section we introduce an estimation procedure of the multivariate MixedTS  based on the distance between the empirical and theoretical characteristic functions . Constraints on tail behavior are considered in order to improve the fitting on tails.
	Before formulating the problem mathematically let us first define the following two quantities: a weighting function given by:
	\begin{equation*}
	\pi\left(t\right)=\left(2\pi\right)^{-N/2}e^{-0.5\left\Vert t\right\Vert ^{2}}
	\end{equation*} 
	and the error term $D_{\tilde{m}}\left(t\right)$ computed on the empirical characteristic function $
  \hat{\varphi}_{\tilde{m}}\left(t\right)=\frac{1}{\tilde{m}}\sum_{j=1}^{\tilde{m}}e^{-i\left\langle t,X_{j}\right\rangle }$ computed on a sample of size $\tilde{m}$ and  the theoretical characteristic function $\varphi\left(t\right)$, defined as:
		\begin{equation*}
	D_{\tilde{m}}\left(t\right):=\hat{\varphi}_{\tilde{m}}\left(t\right)-\varphi\left(t\right).
		\end{equation*} 
		The minimization problem reads:
	\begin{equation}
  \begin{array}{c}
  \underset{\theta\in\Theta}{\min}\int_{\mathbb{R}^{N}}\left\langle D_{\tilde{m}}\left(t\right),\overline{D_{\tilde{m}}\left(t\right)}\right\rangle \pi\left(t\right)\mbox{d}t\\
  \texttt{s.t.}\\
  q^{\star\ \text{emp}}_{i}=q^{\star\ \text{theo}}_{i} \ \ \ \text{for}\ \ \ i=1,\ldots, N \\
  r^{\star\ \text{emp}}_{i}=r^{\star\ \text{theo}}_{i} \ \ \ \text{for}\ \ \ i=1,\ldots, N 
  \end{array}\label{eq:ProbMin1}
  \end{equation}
	where \(\left\langle a,b\right\rangle\) is the inner product between
  vectors \(a\) and \(b\); \(\Theta\) is the set of the MixedTS parameters.
  \(q^{\star\ \text{emp}}_{i}\),
  \(q^{\star\ \text{theo}}_{i}\) determine the left tail
  behavior of the empirical and theoretical marginal distribution while
  \(r^{\star\ \text{emp}}_{i}\),
  \(r^{\star\ \text{theo}}_{i}\) refer respectively to the 
  empirical and theoretical marginal right tail. For the quantities
  \(q^{\star\ \text{emp}}_{i}\),
  \(q^{\star\ \text{theo}}_{i}\),
  \(r^{\star\ \text{emp}}_{i}\),
  \(r^{\star\ \text{theo}}_{i}\) we refer  to Section
  \ref{tailB}.\newline
  The integral in  \eqref{eq:ProbMin1} is evaluated using Monte Carlo simulation since \(t\) can
  be seen as a multivariate Standard Normal random variable. We observe that the objective function in \eqref{eq:ProbMin1} is bounded since the $\int_{\mathbb{R}^{N}}\pi\left(t\right)\mbox{d}t=1$ and the error term $D_{\tilde{m}}\left(t\right)$ is bounded. 
	In the constrained problem \eqref{eq:ProbMin1} we introduce a dynamic penalty to the objective function. Let us first introduce the following two vectors:
	\begin{equation*}
	\Delta\mathbf{q}^{\star}=\left[\left(q^{\star\ \text{emp}}_{1}-q^{\star\ \text{theo}}_{1}\right), \ldots, \left(q^{\star\ \text{emp}}_{N}-q^{\star\ \text{theo}}_{N}\right) \right],
	\end{equation*}
		\begin{equation*}
		\Delta\mathbf{r}^{\star}=\left[\left(r^{\star\ \text{emp}}_{1}-r^{\star\ \text{theo}}_{1}\right), \ldots, \left(r^{\star\ \text{emp}}_{N}-r^{\star\ \text{theo}}_{N}\right) \right].
	\end{equation*}
	The considered penalty function is \(h\left(\Delta\mathbf{q}^{\star},\Delta\mathbf{r}^{\star}\right)\) defined as:
  \[
  h\left(\Delta\mathbf{q}^{\star},\Delta\mathbf{r}^{\star}\right)=\underset{i=1}{\sum^{N}}\left[\left(\Delta{q}_{i}^{\star}\right)^2+\left(\Delta{r}_{i}^{\star}\right)^2\right].
  \] The optimization problem in \eqref{eq:ProbMin1} becomes the following unconstrained optimization:
  \begin{equation}
  \begin{array}{c}
  \underset{\theta\in\Theta}{\min}\int_{\mathbb{R}^{N}}\left\langle D_{\tilde{m}}\left(t\right),\overline{D_{\tilde{m}}\left(t\right)}\right\rangle \pi\left(t\right)\mbox{d}t+\lambda h\left(\Delta\mathbf{q}^{\star},\Delta\mathbf{r}^{\star}\right) \ \ \ \text{for}\ \ \ \lambda >0.\\
  \end{array}
  \label{eq:ProbMin2}
  \end{equation}
A standard approach used when working with dynamic penalty function is to solve a sequence of unconstrained minimization problems:
	 \begin{equation}
  \begin{array}{c}
  L_{l}^{\left(\theta\right)}=\underset{\theta\in\Theta}{\min}\int_{\mathbb{R}^{N}}\left\langle D_{\tilde{m}}\left(t\right),\overline{D_{\tilde{m}}\left(t\right)}\right\rangle \pi\left(t\right)\mbox{d}t+\lambda_l h\left(\Delta\mathbf{q}^{\star},\Delta\mathbf{r}^{\star}\right) \ \ \ \text{for}\ \ \ \lambda_l >0.\\
  \end{array}
  \label{eq:ProbMin3}
  \end{equation}
	where the penality coefficient $\lambda_l$ at each iteration $l$ increases, i.e. $\lambda_l>\lambda_{l-1}$ \citep[see ][ for more details]{eiben2003introduction}. The algorithm stops when $\left|L_{l-1}\left(\theta\right)-L_{l}\left(\theta\right)\right|\leq \epsilon$, for a fixed small $\epsilon$. In this paper we choose a different method where at each iteration $k$ of the \cite{NelderMead65} algorithm the penalty in \eqref{eq:ProbMin2} is updated according to: 
	\begin{equation*}
  \lambda_{k} :=h_{k-1}\left(\Delta\mathbf{q}^{\star},\Delta\mathbf{r}^{\star}\right).
  \end{equation*} In this way instead of solving a sequence of problems defined in \eqref{eq:ProbMin3} we have only one problem to solve. 

\subsection{Numerical Example}

In the previous Section we introduced a methodology for the estimation of a
multivariate MixedTS distribution based on the minimization problem in
\eqref{eq:ProbMin2}. The integral in the objective function is computed through Monte Carlo simulation based on the following approximation:
\begin{equation}
\int_{\mathbb{R}^{N}}\left\langle D_{\tilde{m}}\left(t\right),\overline{D_{\tilde{m}}\left(t\right)}\right\rangle \pi\left(t\right)\mbox{d}t \approx \frac{1}{n_{0}}\sum_{j=1}^{n_{0}}\left\langle D_{\tilde{m}}\left(t_j\right),\overline{D_{\tilde{m}}\left(t_j\right)}\right\rangle
\label{MCint}
\end{equation}where \(t_j=\left(t_{1,j},\ldots,t_{N,j}\right)\) with \(j=1,\ldots, n_{0}\)
are extracted from a \(N\ -\) multivariate Standard Normal and $n_{0}$ refers to the
numbers of points used in the evaluation of integral. \newline
We investigate the behavior of the estimators for $n_{0}=150$ implementing the following steps:
\begin{enumerate}
\item We generate a sample of size $7000$ from a bivariate and a trivariate MixedTS distributions. 
\item Using a boostrap technique with replacement we draw $4000$ samples of size $7000$.
\item For each boostrapped sample we estimate the parameters by solving problem \eqref{eq:ProbMin2}.
\item We analyse the distribution of the obtained estimators, as reported in Table \ref{tabN150} and in Table \ref{tabN150_sim3} respectively for bivariate and trivariate MixedTS distributions.
\end{enumerate}
\begin{table}[h]
\centering
\begin{tabular}{lcccccc}
\hline\hline
 $n_{0}=150$ & true & est & median & sd & I quart & III quart \\
\hline
$\mu_{0,1}$ &  0.0000 &  0.0213 &  0.0042 &  0.0285 & -0.0175 &  0.0266\\
$\beta_1$ &  0.0000 & -0.0057 &  0.0089 &  0.0214 & -0.0061 &  0.0225\\
$m_1$ &  1.0000 &  1.0608 &  1.0617 &  0.2278 &  0.9876 &  1.2024\\
$l_1$ &  1.5000 &  1.3968 &  1.4312 &  0.1988 &  1.2682 &  1.5173\\
$\alpha_1$ &  1.2000 &  1.2390 &  1.2829 &  0.2387 &  1.1505 &  1.4599\\
$\lambda_{p,1}$ &  1.0000 &  1.0955 &  1.1601 &  0.3139 &  1.0277 &  1.3648\\
$\lambda_{m,1}$ &  1.0000 &  1.1865 &  1.1889 &  0.3210 &  1.0180 &  1.4203\\
$\mu_{0,2}$ &  0.0000 &  0.0049 &  0.0000 &  0.0316 & -0.0267 &  0.0288\\
$\beta_2$ &  0.0000 &  0.0084 & -0.0052 &  0.0260 & -0.0249 &  0.0152\\
$m_2$ &  1.0000 &  1.0588 &  1.0000 &  0.1599 &  0.9287 &  1.0881\\
$l_2$ &  1.5000 &  1.4028 &  1.4604 &  0.1644 &  1.3526 &  1.5412\\
$\alpha_2$ &  0.8000 &  0.8059 &  1.0422 &  0.1896 &  0.9022 &  1.1949\\
$\lambda_{p,2}$ &  1.0000 &  1.1884 &  1.0064 &  0.3171 &  0.8367 &  1.2124\\
$\lambda_{m,2}$ &  1.0000 &  1.1774 &  0.9770 &  0.2698 &  0.8436 &  1.1393\\
$n$ &  0.5000 &  0.5146 &  0.5884 &  0.2546 &  0.4295 &  0.7542\\
\hline\hline
\end{tabular}
\caption{Estimated parameters choosing $n_{0} = 150$ in \eqref{MCint} for a bivariate MixedTS distribution with parameters respectively $\theta_1=\left(0,0,1,1.5,1.2,1,1\right)$ and $\theta_2=\left(0,0,1,1.5,0.8,1,1\right)$. \label{tabN150}}
\end{table}

\begin{table}[h]
\centering
\begin{tabular}{lcccccc}
\hline\hline
$n_{0}=150$   & true     & est     & median  & sd     & I quart  & III quart \\
\hline
$\mu_{0,1}$    &    0.0   & -0.0109 & -0.0003 & 0.0148 & -0.0276  & 0.0231 \\
$\beta_1$     &    0.0   & 0.009   &  0.0037 & 0.0124 & -0.0195  & 0.0227 \\
$m_1$         &    1.0   & 0.874   &  1.0079 & 0.0920 &  0.9149  & 1.1898 \\
$l_1$    &    1.5   & 1.589   &  1.4728 & 0.0973 &  1.2910  & 1.6037 \\
$\alpha_1$    &    1.2   & 1.293   &  1.2082 & 0.1220 &  1.0204  & 1.4252 \\
$\lambda_{+,1}$ &  1.0   & 1.137   &  1.1434 & 0.1621 &  0.9317  & 1.4400 \\
$\lambda_{-,1}$ &  1.0   & 1.036   &  1.1609 & 0.1709 &  0.9106  & 1.4735 \\        
$\mu_{0,2}$    &    0.0   & -0.030  &  0.0030 & 0.0164 & -0.0260  & 0.0297 \\
$\beta_2$     &    0.0   & 0.0038  &  0.0059 & 0.0123 & -0.0149  & 0.0261 \\
$m_2$         &    1.0   & 1.1012  &  0.9989 & 0.0795 &  0.8709  & 1.1389 \\
$l_2$    &    1.5   & 1.3658  &  1.4504 & 0.1045 &  1.2822  & 1.6287 \\
$\alpha_2$    &    0.8   & 0.9922  &  0.8842 & 0.0793 &  0.7857  & 1.0449 \\
$\lambda_{+,2}$ &  1.0   & 1.2896  &  1.0929 & 0.1385 &  0.8728  & 1.3232 \\
$\lambda_{-,2}$ &  1.0   & 0.9882  &  1.0910 & 0.1405 &  0.8613  & 1.3337 \\
$\mu_{0,3}$    &    0.0   & -0.0170 & 0.0009  & 0.0185 & -0.0310  & 0.0293 \\
$\beta_3$     &    0.0   & 0.0145  & 0.0005  & 0.0157 & -0.0252  & 0.0278 \\
$m_3$         &    1.0   & 1.0590  & 1.0122  & 0.0915 &  0.8903  & 1.1881 \\
$l_3$    &    1.5   & 1.4448  & 1.4924  & 0.0994 &  1.3272  & 1.6502 \\
$\alpha_3$    &    1.8   & 1.8683  & 1.8181  & 0.1163 &  1.5877  & 1.9538 \\
$\lambda_{+,3}$ &  1.0   & 1.2061  & 1.2104  & 0.1876 &  0.9789  & 1.5956 \\
$\lambda_{-,3}$ &  1.0   & 1.1152  & 1.2462  & 0.1877 &  1.0229  & 1.6222 \\
$n$ & 0.5 & 0.5177 & 0.5414  & 0.1168415 & 0.3591 & 0.7435 \\
\hline
\end{tabular}
\caption{Estimated parameters choosing $n_{0} = 150$ in \eqref{MCint} for a bivariate MixedTS distribution with parameters respectively $\theta_1=\left(0,0,1,1.5,1.2,1,1\right)$, $\theta_2=\left(0,0,1,1.5,0.8,1,1\right)$ and $\theta_3=\left(0,0,1,1.5,1.8,1,1\right)$. \label{tabN150_sim3}}
\end{table}

\section{Conclusion}
\label{conc}
We introduced a new infinitely divisible distribution called multivariate MixedTS. This new distribution is a generalization of the Normal Variance Mean Mixtures.  The flexibility of the multivariate MixedTS distribution is emphasized by means of a direct comparison with the multivariate Variance Gamma, which is a competing model that presents some limits. The multivariate Variance Gamma distribution is not able to capture the dependence structure between components of a random vector, which is important if we work with financial markets data. We showed that using the multivariate $MixedTS-\Gamma$ distribution these limits are overcome, which is due to the presence of the tempering parameters. 
 Taking into account the structure of the new distribution we propose a simulation procedure which exploits the existence of algorithms in literature for the simulation of the Tempered Stable distribution. We also propose an estimation procedure, based on the minimization of a distance between the empirical and theoretical characteristic functions. 
 Results on asymptotic tail behavior of marginals are used as constraints, in the optimization problem, in order to improve tail fitting.  Capturing the dependence of extreme events is helpful in many areas such as in portfolio risk management, in reinsurance or in modeling catastrophe risk related to climate change. The proposed estimation procedure is illustrated through a numerical analysis on simulated data from a bivariate MixedTS. We estimate parameters on bootstrapped samples and investigate their empirical distribution. Finally, some remarks on possible future research starting from this paper are listed below. One can study the multivariate MixedTS considering other  mixing distributions, rather than Gamma. Empirical investigation on the ability of the  proposed distribution in fitting the data in different fields would also be of interest. 
Another important issue would be the study of the efficiency of the estimators resulting from the proposed estimation methodology.

\bibliographystyle{chicago}
\bibliography{PaperBiblio}
 
\appendix
\section{Derivation of higher order moments}
\label{Ap1}

For the general $i^{th}$ element of the random vector $\mathbf{Y}$, the mean is obtained as:%
\[
E\left[  Y_{i}\right]  =E\left[  \mu_{i}+\beta_{i}V_{i}+\sqrt{V_{i}}%
X_{i}\right].
\]From the tower property applied to the conditional expected value we get:%
\begin{align*}
E\left[  Y_{i}\right]   &  =E\left[  \mu_{i}+\beta_{i}V_{i}+E\left[
\sqrt{V_{i}}X_{i}\left\vert \chi\right.  \right]  \right] \\
&  =E\left[  \mu_{i}+\beta_{i}V_{i}+\sqrt{V_{i}}E\left[  X_{i}\left\vert
\chi\right.  \right]  \right]
\end{align*}Observe that $H_{i}:=$ $X_{i}\left\vert \chi\right.$ is a standardized
Tempered Stable, from where we have:%
\[
E\left[  Y_{i}\right]  =\mu_{i}+\beta_{i}E\left[  V_{i}\right]
\]and since $V_{i}\sim\Gamma\left(  l_{i}+n,m_{i}\right)  $%
\[
E\left[  Y_{i}\right]  =\mu_{i}+\beta_{i}\frac{l_{i}+n}{m_{i}}.%
\]First we consider the diagonal elements of the variance-covariance Matrix. For
the $i^{th}$ component of the random vector $\mathbf{Y}$ we have:
\begin{align*}
\sigma_{i}^{2}  &  =E\left[  \left(  Y_{i}-E\left(  Y_{i}\right)  \right)
^{2}\right] \\
&  =E\left[  \left[  \beta_{i}\left(  V_{i}-E\left(  V_{i}\right)  \right)
+\sqrt{V_{i}}X_{i}\right]  ^{2}\right].
\end{align*}Applying the binomial formula we rewrite $\sigma_{i}^{2} $ as:%
\[
\sigma_{i}^{2}=E\left[  \beta_{i}^{2}\left(  V_{i}-E\left(  V_{i}\right)
\right)  ^{2}+V_{i}X_{i}^{2}+2\beta_{i}\left(  V_{i}-E\left(  V_{i}\right)
\right)  \sqrt{V_{i}}X_{i}\right].
\]The linearity property of the expected value allows the identification of three
terms:%
\[
\sigma_{i}^{2}=\beta_{i}^{2}E\left[  \left(  V_{i}-E\left(  V_{i}\right)
\right)  ^{2}\right]  +E\left[  V_{i}X_{i}^{2}\right]  +2\beta_{i}E\left[
\left(  V_{i}-E\left(  V_{i}\right)  \right)  \sqrt{V_{i}}X_{i}\right].
\]We observe that the first term is related to the variance of the $i^{th}$
component of the mixing random vector $V$. Using the tower property, the
second term becomes the expected value of $V_{i}$ while the last expected value is zero.
\[
\sigma_{i}^{2}=\left(  1+\frac{\beta_{i}^{2}}{m_{i}}\right)  \frac{\left(
l_{i}+n\right)  }{m_{i}}%
\]
The covariance $\sigma_{ij}$ between the $i^{th}$ and $j^{th}$ elements of the random vector $\mathbf{Y}$ is obtained as:%
\begin{align*}
\sigma_{ij}  &  =E\left[  \left(  Y_{i}-E\left(  Y_{i}\right)  \right)
\left(  Y_{j}-E\left(  Y_{j}\right)  \right)  \right] \\
&  =E\left[  \left[  \beta_{i}\left(  V_{i}-E\left(  V_{i}\right)  \right)
+\sqrt{V_{i}}X_{i}\right]  \left[  \beta_{j}\left(  V_{j}-E\left(  V_{j}\right)
\right)  +\sqrt{V_{j}}X_{j}\right]  \right] \\
&  =\beta_{i}\beta_{j}E\left[  \left(  V_{i}-E\left(  V_{i}\right)  \right)
\left(  V_{j}-E\left(  V_{j}\right)  \right)  \right]  +E\left[  \sqrt{V_{i}%
}\sqrt{V_{j}}X_{i}X_{j}\right].
\end{align*}Notice that the first term is related to the covariance between the component
of the mixing random vector $V$. The last equality comes from the condition:
\begin{equation*}
E\left[  \beta_{i}\left(  V_{i}-E\left(  V_{i}\right)  \right)  \sqrt{V_{j}%
}X_{j}\right]  =E\left[  \beta_{j}\left(  V_{j}-E\left(  V_{j}\right)  \right)
\sqrt{V_{i}}X_{i}\right]  =0.
\end{equation*}Applying the tower property, we have:%
\[
E\left[  \sqrt{V_{i}}\sqrt{V_{j}}X_{i}X_{j}\right]  = E\left[  \sqrt{V_{i}%
}\sqrt{V_{j}}E\left[  X_{i}X_{j}\left\vert \chi\right.  \right]  \right].
\]We recall that the random variables $H_{i}:=$ $X_{i}\left\vert \chi\right.  $
and $H_{j}:=$ $X_{j}\left\vert \chi\right.  $ are independent for $i\neq j$
and $i,j=1,\ldots,N,$ therefore:%
\[
E\left[  \sqrt{V_{i}}\sqrt{V_{j}}E\left[  X_{i}X_{j}\left\vert \chi\right.
\right]  \right]  =E\left[  \sqrt{V_{i}}\sqrt{V_{j}}E\left[  H_{i}\right]
E\left[  H_{j}\right]  \right]  =0.
\]Finally we compute the covariance :%
\[
\sigma_{ij}=\frac{\beta_{i}\beta_{j}}{m_{i}m_{j}}n.
\]We now compute the term $s_{iii}$ of the skewness-coskewness matrix:%
\begin{align*}
m_{3}  &  =E\left[  \left(  Y_{i}-E\left(  Y_{i}\right)  \right)
^{3}\right] \\
&  =E\left[  \left[  \beta_{i}\left(  V_{i}-E\left(  V_{i}\right)  \right)
+\sqrt{V_{i}}X_{i}\right]  ^{3}\right]
\end{align*}Using the Newton formula we obtain:%
\begin{align*}
m_{3}  &  =E\left[  \left[  \beta_{i}\left(  V_{i}-E\left(  V_{i}\right)
\right)  +\sqrt{V_{i}}X_{i}\right]  ^{3}\right] \\
&  =E\left[  \left[  \beta_{i}\left(  V_{i}-E\left(  V_{i}\right)  \right)
\right]  ^{3}+3\left[  \beta_{i}\left(  V_{i}-E\left(  V_{i}\right)  \right)
\right]  ^{2}\sqrt{V_{i}}X_{i}+3\beta_{i}\left(  V_{i}-E\left(  V_{i}\right)
\right)  \left[  \sqrt{V_{i}}X_{i}\right]  ^{2}+\left[  \sqrt{V_{i}}%
X_{i}\right]  ^{3}\right]
\end{align*}where using the tower property
and the mean of a standardized Tempered Stable we get:
\begin{equation*}
E\left[  3\left[  \beta_{i}\left(  V_{i}-E\left(  V_{i}\right)
\right)  \right]  ^{2}\sqrt{V_{i}}X_{i}\right]  =0
\end{equation*} while
\begin{equation*}
E\left[  3\beta_{i}\left(  V_{i}-E\left(  V_{i}\right)  \right)
\left[  \sqrt{V_{i}}X_{i}\right]  ^{2}\right]  =3\beta_{i}\left[  E\left[
V_{i}^{2}\right]  -E^{2}\left[  V_{i}\right]  \right].
\end{equation*}The quantity $E\left[  \left(  \sqrt{V_{i}}X_{i}\right)  ^{3}\right]  $ is
the third moment of the MixedTS  and can be computed as:
\begin{equation*}
 E\left[  \left(  \sqrt
{V_{i}}X_{i}\right)  ^{3}\right]  =\left(  2-\alpha\right)  \frac{\lambda
_{+}^{\alpha-3}-\lambda_{-}^{\alpha-3}}{\lambda_{+}^{\alpha-2}+\lambda
_{-}^{\alpha-2}}E\left[  V\right].
\end{equation*}Therefore:%
\[
m_{3}=\left[  \left(  2-\alpha\right)  \frac{\lambda_{+}^{\alpha-3}%
-\lambda_{-}^{\alpha-3}}{\lambda_{+}^{\alpha-2}+\lambda_{-}^{\alpha-2}%
}+\left(  3+2\frac{\beta_{i}^{2}}{m_{i}}\right)  \frac{\beta_{i}}{m_{i}}\right]
\frac{\left(  l_{i}+n\right)  }{m_{i}}.%
\]We derive the formula for the central comoments $k_{iiii}$. The first terms $k_{iiii}$ is the fourth central moment of the univariate MixedTS distribution:%
\begin{align*}
m_{4}  &  =E\left[  \left(  Y_{i}-E\left(  Y_{i}\right)  \right)
^{4}\right] \\
&  =E\left[  \left[  \beta_{i}\left(  V_{i}-E\left(  V_{i}\right)  \right)
+\sqrt{V_{i}}X_{i}\right]  ^{4}\right].
\end{align*}Using the binomial formula we have:%
\begin{align*}
m_{4}  &  =\beta_{i}^{4}E\left[  \left(  V_{i}-E\left(  V_{i}\right)
\right)  ^{4}\right]  +4\beta_{i}^{3}E\left[  \left(  V_{i}-E\left(
V_{i}\right)  \right)  ^{3}\sqrt{V_{i}}X_{i}\right]  +\\
&  +6\beta_{i}^{2}E\left[  \left(  V_{i}-E\left(  V_{i}\right)  \right)
^{2}V_{i}X_{i}^{2}\right]  +4\beta_{i}E\left[  \left(  V_{i}-E\left(
V_{i}\right)  \right)  V_{i}^{3/2}X_{i}^{3}\right]  +\\
&  E\left[  V_{i}^{2}X_{i}^{4}\right].
\end{align*}Using the conditional Tempered Stable assumption we have:
\begin{align*}
m_{4}  &  =\beta_{i}^{4}\left(  3+\frac{6}{l_{i}+n}\right)  \frac{\left(
l_{i}+n\right)  ^{2}}{m_{i}^{4}}+6\beta_{i}^{2}\frac{l_{i}+n}{m_{i}^{3}}\left(
l_{i}+n+2\right)  +\\
&  +4\beta_{i}\left(  2-\alpha_{i}\right)  \left(  \frac{\lambda_{+}^{\alpha
-3}-\lambda_{-}^{\alpha-3}}{\lambda_{+}^{\alpha-2}+\lambda_{-}^{\alpha-2}%
}\right)  \frac{l_{i}+n}{m_{i}^{2}}+\left(  3-\alpha_{i}\right)  \left(
2-\alpha_{i}\right)  \left(  \frac{\lambda_{+}^{\alpha-4}+\lambda_{-}%
^{\alpha-4}}{\lambda_{+}^{\alpha-2}+\lambda_{-}^{\alpha-2}}\right)
\frac{l_{i}+n}{m_{i}}.%
\end{align*}

\section{Derivation of the multivariate MixedTS characteristic function}
\label{dmcf}
Let $Y$ be a multivariate MixedTS, its characteristic function
is:%
\begin{equation}
\varphi_{Y}\left(  u\right)  =E\left[  \exp\left(  i\left\langle
u,Y\right\rangle \right)  \right]  \label{ch}%
\end{equation}
Substituting in \ref{ch} the components $Y_{h}$ defined in \eqref{comp:mixedts} we have:%
\begin{equation}
\varphi_{Y}\left(  u\right)  =E\left[  \exp\left(  i%
{\displaystyle\sum\limits_{h=1}^{N}}
u_{h}\ \left(  \mu_{h}+\beta_{h}V_{h}+\sqrt{V_{h}}X_{h}\right)  \right)
\right]
\end{equation}
Applying freezing property for the expected value and considering the
conditional expected value with respect to the $\sigma-$field defined as
$\chi=\sigma\left(  \left\{  G_{h}\right\}  _{h=1,\ldots,N},\ Z\right)  $:%
\begin{equation}
\varphi_{Y}\left(  u\right)  =E\left\{  E\left[  \left.  \exp\left(  i%
{\displaystyle\sum\limits_{h=1}^{N}}
u_{h}\ \left(  \mu_{h}+\beta_{h}V_{h}+\sqrt{V_{h}}X_{h}\right)  \right)
\right\vert {\Large \chi}\right]  \right\}  \label{ch1}%
\end{equation}
Using the condition in \eqref{cond:Hi} and $V_{h}$ measurable with respect to the
$\sigma-$field $\chi$ we obtain:%
\begin{equation}
\varphi_{Y}\left(  u\right)  =E\left\{  \exp\left[  i%
{\displaystyle\sum\limits_{h=1}^{N}}
u_{h}\ \left(  \mu_{h}+\beta_{h}V_{h}\right)  \right]  E\left[  \left.
\exp\left(  i%
{\displaystyle\sum\limits_{h=1}^{N}}
u_{h}\ \sqrt{V_{h}}X_{h}\right)  \right\vert {\Large \chi}\right]  \right\}
\end{equation}
Since $\left.  X_{h}\right\vert {\Large \chi\sim}$ $stdCTS\left(  \alpha
_{h},\ \lambda_{+,h}\sqrt{V_{h}},\lambda_{-,h}\sqrt{V_{h}}\right)  $ and
indicating with $L_{stdCTS}\left( u,\alpha_{h},\ \lambda_{+,h},\lambda
_{-,h}\right)  $ the characteristic exponent of the Standardized Classical
Tempered Stable, we obtain:%
\begin{equation}
\varphi_{Y}\left(  u\right)  =\exp\left[  i%
{\displaystyle\sum\limits_{h=1}^{N}}
u_{h}\ \mu_{h}\right]  E\left\{  \exp\left[  i%
{\displaystyle\sum\limits_{h=1}^{N}}
\left(  u_{h}\ \beta_{h}V_{h}+V_{h}\ L_{stdCTS}\left(  u_{h},\alpha
_{h},\ \lambda_{+,h},\lambda_{-,h}\right)  \right)  \right]  \right\}
\end{equation}
Recalling that $V_{h}\ =G_{h}+a_{h}Z$ as in \eqref{Vh} and using $\Phi_{Z}$ and
$\Phi_{G_{h}}$ for the logarithm of the m.g.f. of $Z$ and $G_{h}$, we obtain
the result in \eqref{charMixedTS}

\end{document}